%% file: generalisedRPtests.tex
\documentclass[noinfoline,11pt,oneside]{article}
\usepackage[breaklinks,colorlinks,citecolor=blue,urlcolor=blue]{hyperref}
\usepackage{fullpage, graphicx,amsmath,amssymb,bbm,amsfonts,mathrsfs,dsfont,amsthm,enumerate,url,color,caption}
\usepackage[font = scriptsize]{subcaption}
\usepackage[round]{natbib}
\usepackage{algorithm}
\usepackage{algcompatible}
\usepackage{hyperref}
\usepackage{subcaption}
\usepackage[dvipsnames]{xcolor}
\theoremstyle{plain}
\newtheorem{thm}{Theorem}
\newtheorem{lem}{Lemma}
\newtheorem{prop}{Proposition}
\newtheorem{cor}{Corollary}

\newtheorem{cond}{Condition}

\theoremstyle{definition}
\newtheorem{defn}{Definition}

\theoremstyle{remark}

\newcommand{\abs}[1]{\left| #1\right|}

\newcommand{\R}{\mathbb{R}}

\newcommand{\Var}{\mathrm{Var}}

\newcommand{\vertiii}[1]{{\left\vert\kern-0.25ex\left\vert\kern-0.25ex\left\vert #1 
    \right\vert\kern-0.25ex\right\vert\kern-0.25ex\right\vert}}

\newcommand{\domin}{d_{\min}}

\newcommand{\hatwj}{\hat w_j}
\newcommand{\wj}{w_j}
\newcommand{\exi}{x_i}

\newcommand{\xone}{x_1}
\newcommand{\xn}{x_n}

\newcommand{\hatR}{{R}}
\newcommand{\hatRi}{{R}_i}

\newcommand{\hatRo}{{R}_{\text{ora}}}
\newcommand{\hatRoi}{{R}_{\text{ora},i}}

\newcommand{\hatRA}{\hat{R}}

\newcommand{\hatRG}{\hat{R}_{G}}
\newcommand{\hatRGi}{\hat{R}_{G,i}}

\newcommand{\Rajen}[1]{{ \color{blue} RDS: #1}}

\newcommand{\algname}[1]{Algorithm }

\newcommand{\del}{{\bf\Delta}}
\newcommand{\Do}{{D_\beta}}
\newcommand{\er}{\text{rem}_1}
\newcommand{\erone}{\text{rem}_{11}}
\newcommand{\ertwo}{\text{rem}_{12}}
\newcommand{\erzero}{\text{rem}_{0}}

\newcommand{\bmu}{{\mu}}

\newcommand{\s}{\frac{1}{n}\sum_{i=1}^n}

\DeclareMathOperator*\argmin{argmin}

\newcommand{\bigCI}{\mathrel{\text{\scalebox{1.07}{$\perp\mkern-10mu\perp$}}}}

\usepackage{cases,multirow}

\newenvironment{myalgo}[3][]{\refstepcounter{myalgocounter}\par\medskip
  \noindent
  \rule{\textwidth}{0.7pt}
	\vskip -0.1cm
  \noindent 
	\noindent \textbf{Algorithm~\themyalgocounter. #1#2} \rmfamily
	\\
	\noindent
  \rule[0.6\baselineskip]{\textwidth}{0.4pt}
	\vskip -.3cm
	\noindent
	}{\medskip 
	}

\newcounter{myalgocounter}[section]
\renewcommand{\themyalgocounter}{
\arabic{myalgocounter}}

\makeatletter
\def\input@path{{/figures/}}
\makeatother

\begin{document}
\title{Goodness-of-fit testing in high-dimensional generalized linear models}
\author{Jana Jankov\'a$^\ast$, Rajen D. Shah$^\ast$, Peter B\"uhlmann$^\dag$ and Richard J. Samworth$^\ast$\\$^\ast$University of Cambridge and $^\dag$ETH Z\"urich}
\date{\today}

\maketitle

\begin{abstract}
We propose a family of tests to assess the goodness-of-fit of a high-dimensional generalized linear model. Our framework is flexible and may be used to construct an omnibus test or directed against testing specific non-linearities and interaction effects, or for testing the significance of groups of variables. The methodology is based on extracting left-over signal in the residuals from an initial fit of a generalized linear model. This can be achieved by predicting this  signal from the residuals using modern flexible regression or machine learning methods such as random forests or boosted trees.
Under the null hypothesis that the generalized linear model is correct, no signal is left in the residuals and our test statistic has a  Gaussian limiting distribution, translating to asymptotic control of type I error.
Under a local alternative, we establish a guarantee on the power of the test. 
We illustrate the effectiveness of the methodology on simulated and real data examples by testing goodness-of-fit in logistic regression models.
 Software implementing the methodology is available in the \textsf{R} package \texttt{GRPtests} \citep{GRPtests}.
\end{abstract}

\section{Introduction}
\label{sec:intro}
In recent years, there has been substantial progress in developing methodology for estimation in generalized linear models 
in high-dimensional settings, where the number of covariates in the model may be much larger than the number of observations.
A standard technique for estimation is the Lasso for generalized linear models \citep{park2007l1}, 
which has a fast implementation in the R package \texttt{glmnet} \citep{glmnet} and is widely used. 
The Lasso enjoys good empirical and theoretical properties for estimation and variable selection, provided that we are searching for a sparse approximation to the regression coefficients in the generalized linear model. 

Once a generalized linear model has been fitted to the high-dimensional data, 
it is important to assess the quality of the fit. 
Literature on testing goodness-of-fit in low-dimensional settings is extensive: we refer to Section \ref{sec:literature} below for an overview. However, the methods typically rely on properties that only hold in low-dimensional settings such as asymptotic linearity and normality of the maximum likelihood estimator, for example. These may fail to hold with an increasing number of covariates in the model; as a consequence it is typically not possible to extend these approaches   in an obvious way to the high-dimensional setting.
This motivates us to develop a new method that may be used for detecting misspecification in the fit of a (potentially high-dimensional) generalized linear model.
\par
To fix ideas, suppose we have data $(x_i, Y_i)_{i=1}^n$ formed of feature vectors $x_i \in \R^p$ and univariate responses $Y_i \in \mathcal{Y} \subseteq \R$. Let us write $X = [\xone,\dots,\xn]^T = [X_1,\dots,X_p] = (X_{ij})_{1\leq i\leq n,1\leq j \leq p}$ for the $n\times p$ design 
matrix and $Y = (Y_1,\dots,Y_n)^T$ for the vector of responses. Consider a generalized linear model \citep{mccullagh} for the data. Specifically, consider the setting where the $Y_i$ are independent conditional on $X$, and the conditional distribution of $Y_i$ only depends on $X$ through the linear combination $x_i^T\beta_0$ for some coefficient vector $\beta_0\in\mathbb R^p$.  In particular, for the conditional expectation, this implies a structure of the form
\[
\mathbb E (Y_i|x_i=x) =: m_0(x) = \mu(x^T\beta_0);
\]
the function $\mu(\cdot)$ is a known inverse link function and $\beta_0$ is unknown. 
Moreover, we assume that $\text{var}(Y_i|x_i=x)=\mu'(x^T\beta_0).$
This structure of the variance arises in generalized linear models derived from exponential families with canonical link functions, such as logistic regression or Poisson log-linear models.

\par 
We will focus on the detection of misspecification in the conditional mean function. 
In a low-dimensional setting, we understand that the model is misspecified in the conditional mean when there does not exist a $\beta_0\in\mathbb R^p$ such that $m_0(x)  = \mu(x^T \beta_0)$. 
In a high-dimensional setting where $p \geq n$, this concept becomes more complicated at first sight; for example, with fixed design points 
$\xone,\dots,\xn$, there always exists 
$\beta_0\in\mathbb R^p$ such that
$m_0(\exi) = \mu(\exi^T\beta_0)$ for all $i=1,\dots,n$,
meaning that the model can never be misspecified. 
However, in a high-dimensional setting, it is impossible to estimate $\beta_0$ consistently without additional structural assumptions. An assumption that is often used, and which we adopt in this paper, is sparsity of the model. Therefore, we address the question of whether a \emph{sparse} model fits well to the observations, or whether a ({sparse})
non-linear model is more appropriate. If we restrict ourselves to sparse models, 
 then misspecification can happen in the same way as in low-dimensional settings, even for fixed design. 
Some of the most important types of misspecification that are of interest in applications are missing nonlinear terms such as quadratic effects or interaction terms. 
Examples of generalized linear models that may be covered by our framework include 
logistic regression, Poisson regression, robust regression (Huber loss, Cauchy loss) and linear regression.

\subsection{Overview of our contributions} \label{sec:overview}
We now briefly outline our strategy for goodness-of-fit testing; a more detailed description is given in Section~\ref{sec:method}. 
Let $\hat{\beta}$ be an estimate of $\beta_0$ derived from a Lasso-penalised generalized linear model (GLM Lasso) fit.
Our starting point is the vector $\hatR$ of Pearson residuals, with $i$-th coordinate 
$$\hatRi:=\frac{Y_i-\mu(\exi^T\hat\beta)}{\sqrt{\mu'(\exi^T\hat\beta)}}, \quad i=1,\ldots,n.$$
%
Now consider taking as a test statistic the scalar product $w^T \hatR$, for some (fixed) unit vector $w \in \R^n$.
 If the generalized linear model were correct, then $w^T \hatR$ would be approximately an average of zero-mean random variables, and under reasonable conditions, should converge to a centred Gaussian random variable. On the other hand, if the model were misspecified, the residuals would contain some signal, and were $w$ to be positively correlated with this signal, the lack of fit should be exposed by the test statistic taking a large value.
 
In the alternative setting, the signal in the residuals may be picked up by more flexible regression methods, such as random forests \citep{breiman2001random} or boosted trees \citep{chen2016xgboost}. However using such flexible regressions to inform the choice of $w$ directly would make $w$ strongly dependent on $\hatR$ even under the null; as such calibration of the resulting test statistic would be problematic. Our approach therefore is to construct $w$ based on an independent auxiliary dataset $(X_A,Y_A)$ (e.g.\ derived through sample splitting) in the following way. We first perform a GLM Lasso fit on the auxiliary dataset to obtain an additional set of residuals. Regressing these residuals back on to the explanatory variables $X_A$ using a flexible regression method, we obtain an estimated regression function $\tilde{f}: \mathbb R^p \to \mathbb R$ that aims to predict the signal in the residuals; we refer to the $n$-fold concatenation of such an $\tilde{f}$ as a \emph{residual prediction function} $\hat f:\mathbb R^{n\times p} \rightarrow \mathbb R^n$. We may then choose $w$ proportional to $\hat{f}(X)$ to give a direction $w$ independent of $Y$. 

One important issue that arises in the high-dimensional setting is that although components of $\hatR$ are close to zero-mean under the null, their bias can drive a substantial shift in the mean of $w^T \hatR$. To prevent this, we replace $w$ with the residuals from a particular weighted (square-root) Lasso regression of $w$ on to $X$. This final step ensures that $w$ is almost orthogonal to the bias in the residuals and as a consequence, the limiting distribution under the null is a centred Gaussian.  A notable feature of our construction is that the asymptotic null distribution is essentially invariant to the residual prediction method used. This can therefore be as flexible as needed to detect the type of mean misspecification we would like to uncover.

We provide a software implementation of our methodology in the \textsf{R} package \texttt{GRPtests} \citep{GRPtests}.


\subsection{Related literature}
\label{sec:literature}
\textbf{High dimensions.}
Our work is related to that of \cite{shah2018} who study goodness-of-fit tests for the linear model. They consider test statistics based on a proxy for the prediction error of a flexible regression method applied to the scaled residuals following a square-root Lasso fit to the data. It is shown that when a Gaussian linear model holds, these residuals depend only weakly on the unknown regression coefficients, motivating calibration of the tests via a parametric bootstrap. As there is no analogue of this result for other generalized linear models, it does not seem possible to extend this approach to our more general setting. Our methodology shares the idea of `predicting' the residuals but, even when we specialize our approach to the Gaussian linear model, differs substantially in the construction of test statistics and the form of calibration.

\par
In recent years, there has been much work on inference and testing in high-dimensional generalized linear models, particularly for the linear model. The work on significance testing includes flexible approaches based on (multiple) sample-splitting \citep{wasserman,pvals, meinshausen2010stability, shah2013variable} which may be combined with other methods.  Another line of work, initiated by \citet{zhang}, proposes a method of de-biasing the Lasso
that can be used for testing significance of variables in the linear regression. \cite{vdgeer13} extend the methodology to generalized linear models; further developments include \citet{stanford1}, \citet{dezeure2017high} and \citet{yu2018confidence}; see also \citet{belloni2014inference}. 
General frameworks for testing  low-dimensional hypotheses about the parameter can be based for example on Neyman orthogonality conditions \citep{vch1,chernozhukov2018double} or on a profile likelihood testing framework \citep{ning2014likelihood}.
In recent work, \cite{zhu2017projection} propose a method for testing more general hypotheses about the parameter vector, 
such as the sparsity level of the model parameter and minimum signal strength. 
\cite{javanmard2017flexible} suggest a procedure to test similar hypotheses about the parameter in linear or logistic regression. 

\textbf{Low dimensions.} 
There are numerous methods for testing goodness-of-fit of a model in low-dimensional settings, especially for the case of logistic regression, which is one of the focuses of this work. The most standard tests are residual deviance and Pearson's chi-squared tests; however, they behave unsatisfactorily if the data contain only a small number of observations for each pattern of covariate values.
There have been a number of strategies to circumvent this difficulty, mainly based on grouping strategies, residual smoothing or 
modifications of Pearson's chi-squared test.

\par
\cite{hosmer1980} proposed two methods of grouping based on ranked estimated logistic probabilities that form groups of equal numbers
of subjects. The disadvantage of these tests (as noted in \cite{le1991goodness}) is that as they are based on a grouping strategy in the space of responses, they lack power to detect departures from the model in regions of the covariate space that yield the same estimated probabilities. For example, a model with a quadratic term may have very different covariate values with the same estimated probability. \cite{tsiatis1980note} circumvents the difficulties faced by Hosmer--Lemeshow tests using a grouping strategy in the covariate space. However, different partitions of the space of covariates may still lead to substantially different conclusions.

\par
\cite{le1991goodness} introduced a test based on residual smoothing using nonparametric kernel methods. Smoothed residuals replace each residual with a weighted average of itself and other residuals that are close in the covariate space. If residuals close to each other are strongly correlated, smoothing does not affect the magnitude of the residuals strongly, while if they are not correlated smoothing will shrink the residuals towards zero.  \cite{su1991lack} proposed a goodness-of-fit test for the generalized linear model based on
a cumulative sum of residuals, which was later adapted by \cite{lin2002model} and a weighted version was proposed in \cite{hosmer2002goodness}.
Another approach based on modifications of Pearson's chi-squared test was studied in \cite{osius1992normal} 
and \cite{farrington1996assessing} who derived a large-sample normal approximation
for Pearson's chi-squared test statistic.

\subsection{Organization and notation}
The rest of the paper is organised as follows. In Section~\ref{sec:method} we motivate and present our goodness-of-fit testing methodology. In Section~\ref{sec:theory}, we study its theoretical properties, providing guarantees on the type I error and power.
In Section~\ref{sec:simulations}, we illustrate the empirical performance of the method on simulated and semi-real genomics data. 
A brief discussion is given in Section~\ref{sec:discuss}.
Proofs are deferred to Section~\ref{sec:proofs} and Appendix~\ref{sec:appendix}.

For a vector $x\in\mathbb R^d$, we let $x_j$ denote its $j$-th entry and write $\|x\|_p:= (\sum_{j=1}^d |x_j|^p)^{1/p}$ for $p \in \mathbb N$, $\|x\|_\infty := \max_{j=1,\dots,d}|x_j|$ and  $\|x\|_0$ for the number of non-zero entries of $x.$
For a matrix $A\in\mathbb R^{n\times p}$, we use the notation $A_{ij}$ or $(A)_{ij}$ for its $(i,j)$-th entry, 
$A_j$ to denote its $j$-th column and we let $\|A\|_\infty:=\max_{i,j}|A_{ij}|$. 
Letting  $G\subseteq \{1,\dots,p\}$, we denote by $A_G$ the matrix containing only columns from $A$ whose indices are in $G$, 
and by $A_{-G}$ the columns of $A$ whose indices are in the complement of $G$.
We use $\Lambda_{\min}(A)$ and $\Lambda_{\max}(A)$ to denote the minimum and maximum eigenvalue of a square matrix $A$.


For sequences of random variables $X_n,Y_n$, we write $X_n=\mathcal O_P(Y_n)$ if $X_n/Y_n$ is bounded in probability and $X_n=o_P(1)$ if $X_n$ converges to zero in probability.  We write $a \lesssim b$ to mean that there exists $C > 0$, which may depend on other quantities designated as constants in our assumptions, such that $a \leq Cb$.  If $a \lesssim b$ and $b \lesssim a$, we write $a \asymp b$.  
Finally, for a function $f:\mathbb R \rightarrow \mathbb R$ and a vector $z=(z_1,\dots,z_n)\in\mathbb R^n$ we will use $f(z)$ to denote the coordinate-wise application of $f$ to $z$, that is $f(z) = (f(z_1),\dots,f(z_n))$.

\section{Methodology: Generalized Residual Prediction tests}
\label{sec:method}

As discussed in Section~\ref{sec:overview}, our Generalized Residual Prediction (GRP) testing methodology relies on an initial fit of the Lasso for generalized linear models, which is defined by
$$\hat\beta:=\argmin_{\beta\in\mathbb R^p} \left\{\s \rho(Y_i,\exi^T\beta) + \lambda \|\beta\|_1\right\}.$$
Here 
$\rho:\mathcal Y \times \mathbb R\rightarrow \mathbb R$ 
is a loss function, usually derived from the negative log-likelihood associated with the model.
Our general framework for goodness-of-fit testing will also assume we have available an auxiliary dataset $(X_A, Y_A) \in \R^{n_A \times p} \times \mathcal{Y}^{n_A}$ independent of $(X, Y)$, sharing the same conditional distribution structure as that of $(X, Y)$. In the rest of the paper, we take $n_A = n$ for simplicity, although this is not needed for our procedures. 
Consider the Pearson-type residuals
\[
\hatR_i = \frac{Y_i - \mu(x_i^T\hat{\beta})}{\sqrt{\mu'(x_i^T\tilde{\beta})}}, \quad i=1,\ldots,n.
\]
Here $\tilde{\beta} \in \R^p$ is an additional estimate of $\beta_0$ that may be computed using the auxiliary dataset, or in certain circumstances may be taken as $\hat{\beta}$ itself: we discuss these two cases in the following sections.
Given the vector $\hatR$ of residuals, 
the basic form of our test statistic is $w^T \hatR$; here $w \in \R^n$ is a direction typically derived using the auxiliary dataset.
We describe in detail the construction of such a $w$ in Section~\ref{subsec:nonlin}, where the goal is general goodness-of-fit testing.

A further modification of the method can allow us to use multiple directions $w$ to test simultaneously for different departures from the null or to aggregate  over different directions derived using flexible regression methods with different tuning parameters.
Given a set $W\subseteq \mathbb R^n$ of direction vectors $w$, our proposed test statistic then takes the form
\[
\sup_{w \in W} w^T \hatR.
\]
\par
We illustrate the use of this more general form of our test statistic for testing the significance  of a group of variables. Such a problem may not immediately seem like goodness-of-fit testing, but is equivalent to testing the adequacy of a model not involving the group of variables in question. We explain how this may be addressed by our framework in Section~\ref{subsec:grouptesting} and describe a wild bootstrap procedure \citep{chernozhukov2013gaussian} to approximate the distribution of the test statistic under the null.

\subsection{Goodness-of-fit testing}
\label{subsec:nonlin}

\par
To motivate our general procedure for goodness-of-fit testing, 
consider the vector $\hatRo$ of Pearson residuals with an oracle variance scaling, whose $i$-th component is given by
$$
\hatRoi 
:= \frac{Y_i -\mu(\exi^T\hat{\beta})}{D_{\beta_0,ii}}, \quad i=1,\dots,n, 
$$
where $D_{\beta_0,ii}^2 := {\mu'(\exi^T\beta_0)}$. 
We may decompose the residuals into noise and estimation error terms by writing
\begin{align} \label{eq:decomp}
\hatRoi &= \varepsilon_i + r_i ,
\end{align}
where $\varepsilon_i:=\bigl\{Y_i  - \mu(\exi^T\beta_0)\bigr\}/D_{\beta_0,ii}$ and $r_i:=\bigl\{\mu(\exi^T\beta_0) - \mu(\exi^T\hat\beta)\bigr\}/D_{\beta_0,ii}$.
If the generalized linear model is correct, then $\mathbb E(\varepsilon_i|\exi) = 0$ {and $\Var(\varepsilon_i | x_i) = 1$}. 
Turning to the remainder term $r_i$, a first-order Taylor expansion of $\mu$ yields the approximation   
$${r_i} \approx D_{\beta_0,ii}\exi^T(\beta_0-\hat\beta).$$
Writing $D_{\beta_0}$ for the diagonal matrix with entries $D_{\beta_0,ii}$ for $i=1,\dots,n$,
 and $\varepsilon:=(\varepsilon_1,\dots,\varepsilon_n)$, we obtain the decomposition
\begin{equation} \label{eq:approx_resid}
\hatRo \approx \varepsilon + D_{\beta_0} X(\beta_0 - \hat{\beta}).
\end{equation}
Consider a unit vector $w \in \R^n$; as discussed in Section~\ref{sec:overview}, this will typically be constructed from an application of a residual prediction method on the auxiliary data.
Our oracle $w^T\hatRo$ 
then satisfies
\begin{equation}
\label{e1}
w^T \hatRo \approx w^T\varepsilon +w^T D_{\beta_0} X(\beta_0 - \hat{\beta}) {=: w^T\varepsilon + \delta}.
\end{equation}
Under suitable conditions on $w$ and on the moments of the errors, the Berry--Esseen theorem should ensure that the pivot term $w^T \varepsilon$ is well approximated by a standardised Gaussian random variable.
To keep the remainder term $\delta$
 in \eqref{e1} under control we can leverage the fact that under the null, we can expect $\|\hat{\beta} - \beta_0\|_1$ to be small. If $w$ satisfies a near-orthogonality condition
 \begin{equation} \label{eq:near_ortho}
 \|X^TD_{\beta_0} w\|_\infty \leq C\sqrt{\log p},
 \end{equation}
for some $C > 0$, then H\"older's inequality will yield $|\delta| \leq C\sqrt{\log p} \|\hat \beta - \beta_0\|_1$, which asymptotically vanishes under suitable conditions on the sparsity of $\beta_0.$

To guarantee the near-orthogonality condition \eqref{eq:near_ortho}, we may use the square-root Lasso \citep{sqrtlasso,sun2012scaled}: for $\lambda_{\mathrm{sq}} > 0$, let
\[
\hat{\beta}_{\text{ora-sq}} := \argmin_{\beta\in\mathbb R^p} \left\{\frac{1}{\sqrt{n}}
\|D_{\beta_0}({\hat f} (X) - X \beta)\|_2 + \lambda_{\mathrm{sq}}\|\beta\|_1\right\}. 
\]
The Karush--Kuhn--Tucker (KKT) conditions for the convex programme imply that the resulting vector of scaled residuals,
\[
w_{\text{ora}} := \frac{D_{\beta_0} ({\hat f} (X) - X \hat\beta_{\text{ora-sq}})}{\|D_{\beta_0}({\hat f} (X) - X \hat\beta_{\mathrm{ora-sq}})\|_2},
\]
satisfies the near-orthogonality property $\|X^T D_{\beta_0} w_{\text{ora}} \|_\infty \leq C\sqrt{\log p}$ when $\lambda_{\mathrm{sq}} = C \sqrt{\log p/n}$. Note that in performing this square-root Lasso regression, we are not assuming that $\hat{f}$ is well-approximated by a sparse linear combination of variables: we are simply exploiting the stationarity properties of the solution to the square-root Lasso optimisation problem\footnote{In principle, there is a possibility that we obtain a degenerate solution with ${\hat f} (X)= X \hat\beta_{\text{ora-sq}}$. However, we can obseve directly whether or not this occurs, and have never seen this happen in any of our numerical experiments.}. 


From the reasoning above, we conclude that, under appropriate conditions, a simple test based on the asymptotic normality of $w_{\text{ora}}^T \hatRo$ will keep the type I error under control. In order to create a version of the test statistic that does not require oracular knowledge of $D_{\beta_0}$, we may replace this quantity with variance estimates based either on $\hat{\beta}$ or on an estimate derived from the auxiliary data; we use the latter approach as this simplifies the analysis. The overall procedure is summarised in Algorithm 1 below.



\par
%

\par

\begin{myalgo}{\bf Goodness-of-fit testing.}{}
\small
\label{rptest}
 \textbf{Input:} 
sample $(X,Y)\in\mathbb R^{n\times p} \times \mathcal Y^n$;
auxiliary sample $(X_A,Y_A)\in\mathbb R^{n\times p} \times \mathcal Y^n; \lambda, \lambda_A, \lambda_{\mathrm{sq}}>0.$

\begin{enumerate}[1:]
\item
\textbf{Estimation:}
Fit a GLM Lasso to $(X,Y)$ and $(X_A,Y_A)$ (with tuning parameters $\lambda, \lambda_A$ respectively) yielding estimators $\hat\beta$ and  $\hat\beta_A$, respectively.
\item
\textbf{Residual prediction:} Compute the residuals 
${Y_A - \bmu(X_A\hat\beta_A)}$ 
and 
fit a flexible regression method of these residuals versus $X_A$ to obtain a prediction function 
${\hat f}:\mathbb R^{n\times p} \rightarrow \mathbb R^n.$
%
%
\item
\textbf{Near orthogonalization:}
Construct the diagonal weight matrix $\hat{D}_A^2:= \text{diag(}\bmu'(X \hat\beta_A))$ and compute an approximate projection of the prediction $\hat D_A {\hat f} (X)$ onto the column space of $\hat D_A X$:
\begin{equation}
\label{sqrt.lasso}
\hat\beta_{\mathrm{sq}}:=
\argmin_{\beta\in\mathbb R^p} \left\{\frac{1}{\sqrt{n}}
\|\hat D_A({\hat f} (X) - X \beta)\|_2 + \lambda_{\mathrm{sq}}\|\beta\|_1\right\}. 
\end{equation}
Define a direction  
\begin{equation} \label{eq:w_hat_def}
\hat w_A:= \frac{\hat D_A ({\hat f} (X) - X \hat\beta_{\mathrm{sq}})}{\|\hat D_A ({\hat f} (X) - X \hat\beta_{\mathrm{sq}})\|_2}.
\end{equation}
\item
\textbf{Test statistic:} Compute  
the residual vector 
$\hatRA := \hat D_A ^{-1}(Y - \bmu(X\hat\beta))$ and let 
 $T:=\hat w_A^T  \hatRA$.
\end{enumerate}
\textbf{Output:}
$p_{\text{value}}= 1 - \Phi(T)$\\
\rule[0.6\baselineskip]{\textwidth}{0.4pt}

\end{myalgo}

%
%

\par

In practice, the auxiliary dataset $(X_A,Y_A)$ would be obtained through sample splitting.  The effect of the randomness induced by such a split can be mitigated using methods designed to aggregate over multiple sample splits, as studied for instance in \citet{pvals}.
\par

%

%



\subsection{Group testing}
\label{subsec:grouptesting}
Our framework of residual prediction tests also encompasses significance testing of groups of regression coefficients in a generalized linear model. Suppose that we wish to test $H_0:\beta_G =0$ for a given group $G\subseteq \{1,\dots,p\}$. We first form the vector $\hatRG$ of residuals based on a GLM Lasso fit of $Y$ on $X_{-G}$. Then, rather than constructing a single direction $w$ using an auxiliary dataset, we can use multiple directions given by the columns of $X_G$. Specifically, we use the test statistic $\max_{j \in G} |\hat w_j^T \hatRG|$ where $\hat w_j$ is given by the scaled residuals of the weighted square-root Lasso regressions of $X_j$ on to $X_{-G}$.

Note that under the null, $X_G$ will be independent of the noise $\varepsilon$ \eqref{eq:decomp}, and so sample splitting is not necessary in this case to mitigate the potentially complicated dependence of the directions and residuals $\hatRG$. The limiting distribution of the test however will not be Gaussian due to the maximisation over multiple directions. Instead, we argue that $\max_{j \in G} |\hat{w}_j^T \hatRG| \approx \max_{j \in G} |\hat{w}_j^T \varepsilon|$ and then use a wild bootstrap procedure to approximate the distribution of this latter quantity.
The overall procedure is summarised in Algorithm 2 below.
 
\begin{myalgo}{Group Test.}{}
\small
\label{grptest}
\noindent
\textbf{Input:} Group $G\subseteq \{1,\dots,p\};$ sample $(X,Y)\in\mathbb R^{n\times p }\times \mathbb R^n;$ 
$B\in \mathbb N$; $\lambda, \lambda_{\text{nw}}>0.$
\begin{enumerate}[1:]
\item
Fit a GLM Lasso to $(X_{-G},Y)$ with a tuning parameter $\lambda$ to obtain an estimator $\hat\beta_{-G}\in\mathbb R^{p-|G|}$. 
Let $\hat D^2 := \text{diag}(\mu'(X_{-G}\hat\beta_{-G})).$ Compute the vector of residuals
$\hatRG := \hat D^{-1}(Y - \bmu(X_{-G}\hat\beta_{-G})).$
\item
For each $j\in G,$ 
 compute the nodewise regression estimator 
$$\hat \gamma_j:=\argmin_{\gamma\in\mathbb R^{p-|G|}}
\frac{1}{\sqrt{n}}\|\hat D(X_j-X_{-G}\gamma)\|_2 
+ 
\lambda_{\text{nw}} \|\gamma\|_1,$$
and let 
$$\hatwj:= \frac{\hat D(X_j-X_{-G}\hat\gamma_j)}{\|\hat D(X_j-X_{-G}\hat\gamma_j)\|_2}.$$

\item
Evaluate the test statistic $T:=\max_{j\in G} |{\hatwj^T \hatRG}|.$ 
\item
For $b=1,\dots,B$ generate independent random variables $e_1^b,\dots,e_n^b\sim \mathcal N(0,1)$ and let 
$$T^b := \max_{j\in G}\abs{\sum_{i=1}^n  \hat w_{j,i}  \hatRGi e_i^b },$$
where $\hat w_{j,i} $ and $\hatRGi$ are the $i$-th entries of $\hat w_j$ and $\hatRG$, respectively.
\item
Calculate the p-value 
\[
p_{\text{value}}:= \frac{1}{B+1} \biggl(1+\sum_{b=1}^B \mathds{1}_{\{T^b \geq T\}}\biggr).
\]
\end{enumerate}
\textbf{Output:} $p_{\text{value}}$ \\
\rule[0.6\baselineskip]{\textwidth}{0.4pt}
\end{myalgo}

%

Our Algorithm \ref{grptest} is similar to the de-biased Lasso for generalized linear models \citep{vdgeer13}.
 The main difference however is that the de-biased Lasso aims to ensure the directions $\hat{w}_j$ are almost orthogonal to $X_{-j}$, whereas we only impose near-orthogonality with respect to $X_{-G}$. Thus for large groups $G$, more of the direction of $X_G$ is preserved in the $\hat{w}_j$, which typically leads to better power of the test.

\section{Theoretical guarantees}
\label{sec:theory}
In this section we provide theoretical guarantees for the tests proposed in Algorithms~\ref{rptest} and \ref{grptest}. 
We consider an asymptotic regime with the sample size $n$ tending to infinity and the number of parameters $p=p_n$ growing as a function of $n$. 

\subsection{Size of the test}
In the following sections, we show that under the null hypothesis, the size of the test is asymptotically correct. We explore goodness-of-fit testing in Section \ref{subsec:rptest} and group testing in Section~\ref{subsec:group.testing}.

\subsubsection{Goodness-of-fit testing}
\label{subsec:rptest}
Here we show that our test statistic has a Gaussian limiting distribution and we establish a bound on the type I error of the test. In this section, we condition on the design and the auxiliary dataset.  Our result makes use of the fact that when the model is well specified, the GLM Lasso performs well in terms of estimation. 
Specifically, under certain conditions, 
it holds with high probability that $\hat\beta\in\Theta(\lambda,\beta_0,X)$,
where $\Theta(\lambda,\beta_0,X)$ is a local neighbourhood of $\beta_0$ defined by
$$
\Theta(\lambda,\beta_0,X):= 
\left\{\vartheta\in\mathbb R^p: \|\vartheta-\beta_0\|_1 \leq s\lambda, \|X(\vartheta-\beta_0)\|_2^2/{n} \leq s\lambda^2 \right\},
$$
with $s:=\|\beta_0\|_0$ as the number of non-zero entries of $\beta_0$; see for example \citet[][Corollary~6.3]{hds}. Sufficient conditions for this to occur include $\lambda \asymp \sqrt{\log p/n}$, $s=o(n/\log p)$ and further conditions on the tail behaviour of the errors $Y_i - \mu(\exi^T\beta_0)$, the design matrix $X$ and the link function, as detailed below.

\begin{cond}
\label{model}
Assume that 
$\mathbb E (Y_i|\exi = x) = \mu(x^T\beta_0)$ and that $\emph{var}(Y_i|\exi=x)=\mu'(x^T\beta_0)$, that the inverse  link function $u\mapsto \mu(u)$ is differentiable, $u\mapsto \mu'(u)$ is Lipschitz with constant $L$, and that $\mu'(u)>0$ for all $u\in\mathbb R$. Suppose moreover that the weights satisfy
$\min_i D_{\beta_0,ii}^{}\geq \domin$ for some constant $\domin>0.$ 
Assume that $\mathbb E\bigl\{|Y_i - \mu(\exi^T\beta_0)|^3/D_{\beta_0,ii}^3\bigm|X\bigr\} \leq C_{\varepsilon}$ for some constant $C_{\varepsilon}>0$, that $\max_{i=1,\dots,n}\|\exi\|_\infty \leq K_X$ for some $K_X\geq 1$ and 
that $12\domin^{-2} LK_X s\lambda_A\leq 1.$
\end{cond}

Condition \ref{model} is satisfied for generalized linear models with canonical links under mild additional conditions.
For example, in the case of logistic regression, the condition on the weights is satisfied if the class probability $\pi_0(x) = \mathbb{P}(Y_i = 1|\exi =x)$ is bounded away from zero and one. Boundedness of the design (along with other conditions, including $12d_{\min}^{-2} LK_X s\lambda_A\leq 1$) guarantees that the weights can be consistently estimated.  
For our result below, it is convenient to introduce the shorthand notation ${Z_A}:= (X, X_A,Y_A)$.

\begin{thm}
\label{glmnonlinH0}
Consider Algorithm \ref{rptest} with tuning parameters $\lambda,\lambda_A,\lambda_{\mathrm{sq}}>0$.  Assume that Condition~\ref{model} is satisfied, that $\hat\beta_A\in \Theta(\lambda,\beta_0,X_A)$ and let
\begin{equation*}
\label{consistency.lasso}
\delta := \mathbb{P}\bigl(\hat\beta \notin \Theta(\lambda,\beta_0,X)\bigm|X\bigr).
\end{equation*}
Then there exists a constant\footnote{Here and below, the constants in the conclusions of our results may depend upon quantities introduced as constants in the relevant conditions for these results.} $C>0$ such that whenever $Z_A$ satisfies $\hat{f}(X) \neq X\hat{\beta}_{\mathrm{sq}}$, we have for any $z\in\mathbb R$ that
\begin{equation}
\label{T1.bound}
 \left|\mathbb{P}\left(T  \leq  z|{Z_A}\right) - \Phi_{}(z)\right|
 \leq
\delta + C\bigl\{\lambda_{\mathrm{sq}} \sqrt{n} s\lambda + \|\hat w_A\|_\infty {s}(\lambda^2 + \lambda_{A}^2) {n}  + K_Xs\lambda_A+ 
 \|\hat w_A\|_\infty \bigr\}, 
\end{equation}
where $\Phi$ denotes the standard normal distribution function.

\end{thm}

For the asymptotically optimal choice of tuning parameters $\lambda\asymp\lambda_{A}\asymp K_X\sqrt{\log p/n}$ and
 $\lambda_{\mathrm{sq}}\asymp \sqrt{\log p/n},$
the bound in Theorem \ref{glmnonlinH0} reduces to
$$
 \left|\mathbb{P}\left(T  \leq  z|{Z_A}\right) - \Phi_{}(z)\right|
 =\mathcal O_{P}\left(
\delta +  \frac{s(K_X\log p+K_X^2\sqrt{\log p})}{\sqrt{n}} + \|\hat w_A\|_\infty {s}K_X^2\log p  
  \right).$$
  We now discuss the terms on the right-hand side of \eqref{T1.bound}. The terms $\lambda_{\mathrm{sq}} \sqrt{n} s\lambda$ and $K_Xs\lambda_A$ arise from bounding the bias term (near-orthogonalization step) and from bounding the weights, respectively.  The presence of the $\|\hat w_{A}\|_\infty$ term in the bound stems from the contribution of each individual component $\hat w_{A,i}$ to the variance of the pivot term and that of the higher-order terms omitted in \eqref{eq:approx_resid}, which create a bias in the distribution of the test statistic.

  To provide some intuition on the size of $\|\hat w_{A}\|_\infty$, recall that $\hat{w}_A$ is a vector in $\mathbb{R}^n$ with $\|\hat{w}_A\|_2 = 1$, so we may hope for $\|\hat w_{A}\|_\infty$ to be small; in fact, it can be shown in certain settings, and under additional technical conditions, that $\|\hat w_{A}\|_\infty = \mathcal O_P(\log n/\sqrt{n})$. 
  We also remark that we observe $\hat{w}_A$, and if the size of its $\ell_\infty$-norm is a concern, then we can modify the square-root Lasso objective to control it explicitly. Indeed, consider setting
\[
(\tilde{\beta}_{\mathrm{sq}}, \tilde{\eta}_{\mathrm{sq}}):=\argmin_{(\beta, \eta) \in \R^p \times \R^n} \left\{\frac{1}{\sqrt{n}} \|\hat{D}_A(\hat{f}(X) - X\beta) - \eta\|_2 + \lambda_{\mathrm{sq}}\|\beta\|_1 + \lambda_\eta\|\eta\|_1 \right\}
\]
and let
\[
\tilde{w}_A := \frac{\hat{D}_A(\hat{f}(X) - X\tilde{\beta}_{\mathrm{sq}}) - \tilde{\eta}_{\mathrm{sq}}}{\|\hat{D}_A(\hat{f}(X) - X\tilde{\beta}_{\mathrm{sq}}) - \tilde{\eta}_{\mathrm{sq}}\|_2}.
\]
Then the KKT conditions of the optimisation problem imply in particular both that a near-orthogonality condition similar to \eqref{eq:near_ortho} is satisfied for suitable $\lambda_{\mathrm{sq}}$, and also that $\|\tilde{w}_A\|_\infty \leq \sqrt{n}\lambda_\eta$.  Our empirical results in Section \ref{sec:simulations} however suggest that in practice  $\|\hat w_{A}\|_\infty$ typically satisfies the necessary constraint and therefore we propose to use the simpler standard square-root Lasso without the above modifications.

\subsubsection{Group testing}
\label{subsec:group.testing}
In this section, we derive theoretical properties for the group testing procedure proposed in Algorithm \ref{grptest}.
Since we do not use sample splitting, we cannot directly apply the arguments of Theorem~\ref{glmnonlinH0},
as the direction $\hat w_j$ 
depends on $\hat\beta_{-G}$ via the weights $\hat D=D_{\hat\beta_{-G}}$.
In order to understand this dependence, here we consider the setting of random bounded design and assume the response--covariate pairs $(Y_i, x_i)$ are all independent and identically distributed. 


We aim to use the multiplier bootstrap procedure \citep{chernozhukov2013gaussian} to estimate the distribution of the test statistic $\max_{j\in G}|T_j|$ as described in Algorithm \ref{grptest}, but we first summarize  a preliminary result which shows that, under appropriate conditions, $T_j$ can be asymptotically approximated by the zero-mean average $\wj^T \varepsilon$. Here we 
define
$w_j := D_{\beta_0} (X_j - X_{-G}\gamma_{0,j}) /(\sqrt{n} \tau_j)$, where 
$$\tau_j^2:= \frac{1}{n}\mathbb E\|D_{\beta_0} (X_j - X_{-G}\gamma_{0,j})\|_2^2,$$
 and $\gamma_{0,j}$ is the population version of $\hat\gamma_j$ from Algorithm \ref{grptest}; i.e.
$$\gamma_{0,j}:= \argmin_{\gamma\in\mathbb R^{p-|G|}} \frac{1}{n}\mathbb E \|D_{\beta_0} (X_j - X_{-G}\gamma_{})\|_2^2.$$
 Recall that
$$\varepsilon= D_{\beta_0}^{-1}\bigl(Y-\mu(X\beta_0)\bigr).$$
In order to guarantee consistency of $\hat\gamma_j$ in Algorithm \ref{grptest}, we will introduce the additional requirement that $\gamma_{0,j}$ is sparse.
Denote by $\beta_{0,-G}\in \mathbb R^{p-|G|}$ the subset of components of $\beta_0$ corresponding to indices in $G^c.$ 
We also define
$$\Theta_{-G}(\lambda,\beta_0):=
\left\{ \vartheta\in \mathbb R^{p-|G|}: \|\vartheta - \beta_{0,-G}\|_1 \leq s\lambda, 
\|X_{-G}(\vartheta - \beta_{0,-G})\|_2^2 / n \leq s\lambda^2 \right\}.$$

\begin{cond}
\label{grptest.cond}
\begin{itemize}
\item[]
\item[(i)]
Let $\eta_i := Y_i - \mu(\exi^T\beta_{0})$ and assume that there exist constants $c_1,c_2,c_3>0$ such that  
$$\mathbb E( e^{\eta_i^2/c_{1}^2}|X)\leq c_2\;\; \text{ and }\;\;
\mathbb E(\eta_i^2|X) \geq c_3,$$ 
for all $i=1,\dots,n$.
\item[(ii)]
  There exists $K\geq 1$ such that $\|X\|_\infty \leq K$ and $\max_{j\in G}\|X_{-G}\gamma_{0,j} \|_\infty \leq K$.
\item[(iii)]
For some $\delta>0$ and all $\beta\in\mathbb R^p $ satisfying $\|\beta-\beta_0\|_1\leq \delta$ it holds that $c_0\leq \mu'(x^T\beta) \leq C_0$ for some constants $c_0,C_0>0$ and all $x \in \mathbb{R}^p$ with $\|x\|_\infty \leq K$.
\item[(iv)]
Denoting $\Sigma_{0} := \mathbb E X_{}^T D_{\beta_0}^2 X_{} /n $, we have $1/\Lambda_{\min}(\Sigma_{0})\leq C_{e}$
and $\|\Sigma_{0}\|_\infty \leq C_{e}$ for some constant $C_e>0.$
\item[(v)]
We have  $\max_{j\in G}\|\gamma_{0,j}\|_0\leq s$, $\|\beta_0\|_0\leq s$ and there exists a sequence $(a_n)$ with $a_n \rightarrow 0$ and $K^3 {s \log p}/{\sqrt{n}} \leq a_n$.  
\end{itemize}
\end{cond}

\begin{prop}\label{grouptest.prop} 
Assume that Conditions \ref{model} and \ref{grptest.cond} are satisfied and assume that $\hat\beta_{-G}$ satisfies 
\begin{equation}
\label{consistency.lasso.group}
 \mathbb{P}\bigl(\hat\beta_{-G} \in \Theta_{-G}(\lambda,\beta_0)\bigr)\leq 1/p.
\end{equation}
Consider  $\hatwj,j\in G$ (assumed to be non-degenerate) as defined in Algorithm \ref{grptest} 
 with tuning parameters $\lambda\asymp \sqrt{\log p/n}$ and $ \lambda_{\emph{nw}} \asymp K\sqrt{\log p/n}$.
Assume that the null hypothesis  $H_0:\beta_{0,G} = 0$ holds. 
Then there exists a constant $C>0$ such that with probability at least $1 - 3/p$, we have 
$$\max_{j \in G} |T_j - \wj^T \varepsilon| \leq C K^3 \frac{s \log p}{\sqrt{n}}.$$
\end{prop} 

Using Proposition \ref{grouptest.prop} and the results of \cite{chernozhukov2013gaussian}, we can show that the quantiles of
$\max_{j\in G} |T_j| $ can be approximated by the quantiles of 
$T^b:= \max_{j\in G}|\sum_{i=1}^n \hat w_{j,i}  \hatRGi e_i^b|$ where $e_1^b,\dots,e_n^b$ are independent $ \mathcal N(0,1)$ random variables.
We only need to guarantee that $T_j$ is well-approximated by $w_j^T\varepsilon$ and we pay a price of $\log |G|$ for testing $|G|$ hypothesis simultaneously, where $|G|$ denotes the cardinality of $G$.
 
Define the $\alpha$-quantile of $T^b$ conditional on $(\exi,Y_i)_{i=1}^n$ by
$$c_{T^b}(\alpha) := \inf\{t\in \mathbb R: \mathbb{P}_e(T^b \leq t) \geq \alpha\},$$
where $\mathbb{P}_e$ is the probability measure induced by the multiplier variables $(e_i^b)_{i=1}^n$ holding $(\exi,Y_i)_{i=1}^n$ fixed.

\begin{thm}
\label{mbcor}
Assume the conditions of Proposition \ref{grouptest.prop} 
and that there exist constants $C_2,c_2>0$ such that
\begin{equation}
\label{raten} 
\max\left\{K^3 \frac{s \log p}{\sqrt{n}} \log (2|G|) + 4/p,\;\; K^{4}\log(2|G|n)^7 /n 
\right\} \leq C_2 n^{-c_2}.
\end{equation}
Then there exist constants $c,C>0$ 
such that
$$\sup_{\alpha\in(0,1)} \biggl|\mathbb{P}\biggl(\max_{j\in G}|T_j| < c_{T_b} (\alpha) \biggr) - \alpha\biggr|\leq C n^{-c}.$$
\end{thm}
Theorem \ref{mbcor} shows that if the generalized linear model is correct and the null hypothesis $\beta_{0,G} = 0$ holds, 
then Algorithm \ref{grptest} produces an asymptotically valid p-value for testing this hypothesis.

\subsection{Power analysis for goodness-of-fit testing}
\label{sec:power}

The choice of $w$ as postulated in Theorem \ref{glmnonlinH0} guarantees that the type I error for goodness-of-fit testing
stays under control. We now provide guarantees on the local power of the test.
To this end, let us suppose that the true model has conditional expectation function $m_0(x) = \mu\bigl(x^T\beta_0 + g_0(x)\bigr)$. Here $g_0$ represents a small nonlinear perturbation of the linear predictor $x^T\beta_0$. Our aim here is to understand how this propogates through to the distribution of our test statistic. We will suppose that the perturbation is small enough that GLM Lasso estimates lie with high probability within local neighbourhoods of $\beta_0$.
Let us first provide some intuition on the expected value of the test statistic under model misspecification. Writing $f_0(x) = x^T\beta_0 + g_0(x)$, the expectation of the theoretical residuals $\varepsilon = D^{-1}_{\beta_0}\bigl(Y - \mu(X\beta_0)\bigr)$ is given by 
$$w_0:=\mathbb E \varepsilon = \mathbb E D^{-1}_{\beta_0}\bigl(\mu(f_0(X)) - \mu(X\beta_0)\bigr).$$
As argued in Section \ref{sec:method}, the oracular test statistic $w^TR_{\text{ora}}$ can be approximated by the scalar product $w^T \varepsilon$.  Therefore, in order to obtain good power properties, we should seek to construct a direction $w$ so as to maximize $\mathbb Ew^T\varepsilon$. The oracular choice $w:= w_0 / \|w_0\|_2$ yields by a Taylor expansion the approximation 
$$\mathbb Ew^T\varepsilon = \|w_0\|_2 \approx \|D_{\beta_0}(f_0(X) - X\beta_0)\|_2=\|D_{\beta_0}g_0(X)\|_2.$$
We therefore see that the test statistic behaves in expectation as a weighted $\ell_2$-norm of the nonlinear term $g_0(X)$.
We now provide a theoretical justification which can be used for local asymptotic guarantees on the power of our method. 
We introduce the following conditions which are modifications of Condition~\ref{model} to account for the case when the model is misspecified. 

\begin{cond}
\label{model2}
Assume that
$\mathbb E (Y_i|\exi = x) = \mu(f_0(x))$, the inverse link function $u\mapsto \mu(u)$ is differentiable, $u\mapsto \mu'(u)$ is Lipschitz with constant $L$, and $\mu'(u)>0$ for all $u\in\mathbb R$. Suppose moreover that the weights satisfy
$\min_i D_{\beta_0,ii}^{} \geq \domin$ for some constant $\domin > 0$.
Assume that
$\mathbb E \bigl(|Y_i - \mu(f_0(\exi))|^3/D_{Y,ii}^3\bigm|X_i\bigr) \leq C_{\varepsilon}$ for a constant $C_{\varepsilon} > 0$, where
we denote $D_Y^2 := \emph{cov}(Y|X).$
Let $\max_{i=1,\dots,n}\|\exi\|_\infty \leq K_X$ for some $K_X \geq 1$
and assume that $12{\domin^{-2}}{L K_X s\lambda }\leq 1$ 
and $|D_{Y,ii}^2  D_{\beta_0,ii}^{-2} - 1| \leq 2\domin^{-2} L K_X s\lambda$.
\end{cond}

\begin{thm}
\label{glmnonlin}
Consider Algorithm \ref{rptest} with tuning parameters $\lambda,\lambda_A,\lambda_{\mathrm{sq}}$. 
Assume Condition \ref{model2}, that $\hat\beta_A\in \Theta(\lambda,\beta_0,X_A)$ and let
\begin{equation}
\label{consistency.lasso2}
\delta := \mathbb{P}(\hat\beta \notin \Theta(\lambda,\beta_0,X)|X).
\end{equation}
Then there exists a constant $C>0$ such that, whenever $Z_A$ is such that $\hat{f}(X) \neq X\hat{\beta}_{\mathrm{sq}}$, we have for any $z\in\mathbb R$ that
\begin{equation}
\label{T2.bound}
\left|\mathbb{P}\left(T  - \del < z|{Z_A}\right) - \Phi(z)\right| 
\leq
\delta + C\bigl\{\lambda_{\mathrm{sq}} \sqrt{n} s\lambda + \|\hat w_A\|_\infty {s}(\lambda^2 + \lambda_{A}^2) {n}  + K_Xs\lambda_A+ 
 \|\hat w_A\|_\infty \bigr\}
,
\end{equation}
 where
$$\del := {\hat w_A^T \bigl\{\mu\bigl(f_0(X)\bigr) - \mu(X\beta_0)\bigr\}}.$$
\end{thm}

%

Under the null hypothesis, we have $\del =0 $ and $D_Y =D_{\beta_0}$; thus we recover the result of Theorem~\ref{glmnonlinH0}.
The departure of the model from the null hypothesis is captured by $\del$. 
Hence the theorem shows that to detect departures from the null, the direction $w$ must be ``correlated'' with the signal that remains in the residuals under misspecification, namely  $\mu(f_0(X))-\mu(X\beta_0)$. 
 Under a local alternative, e.g. $f_0(X) = X\beta_0 + g_0(X)/\sqrt{n}$, where $\|g_0(X)\|_2=1$, we have
$\del \asymp 1$.

Theorem \ref{glmnonlin} relies on rates of convergence of the ``projected'' estimator $\hat\beta$ in \eqref{consistency.lasso2} when the model is misspecified. Oracle inequalities for Lasso-regularized estimators in high-dimensional settings have been well explored; we refer to \cite{hds} and the references therein. If there is misspecification, we hope that the projected estimator behaves as if it knows which variables are relevant for a linear approximation of the possibly nonlinear target $f_0.$
In Appendix \ref{subsec:orac}, we summarize how misspecification affects estimation of the best linear approximation, based on the approach of \cite{hds}. These results guarantee that under a local alternative, the Lasso for generalized linear models still satisfies the condition
$$\mathbb{P}\bigl(\hat\beta \in \Theta(\lambda,\beta_0,X)|X\bigr)\rightarrow 1.$$

\subsection{Consequences for logistic regression}
\label{sec:cons.logistic}
In this section we show how our general theory applies to the problem of goodness-of-fit testing for logistic regression models. 
We take $\mathcal Y=\{0,1\}$ and assume  $(Y_i |x_i = x) \sim\text{Bernoulli}(\pi_0(x))$.
Define
$$f_0(x) := \log \left(\frac{\pi_0(x)}{1-\pi_0(x)}\right),$$
that is, $\mathbb E(Y_i|\exi=x)=\pi_0(x) = \mu(f_0(x)),$ for the inverse  link  function $\mu(u)=1/(1+e^{-u}).$
The function $f_0$ may be potentially nonlinear in $x.$
The $\ell_1$-regularized logistic regression estimator is
$$\hat\beta:= \argmin_{\beta\in\mathbb R^p}  \s \bigl\{-Y_i x_i^T\beta + d(x_i^T\beta) +\lambda \|\beta\|_1\bigr\},$$
where $d(\xi):= \log (1+e^{\xi})$.  Let $\beta_0 \in \mathbb{R}^p$ 
be the best approximation obtained by a GLM \citep[][Section~6.3, p.~115]{hds}. 
We define $S:=\{j:\beta_{0,j} \neq 0\}$ and $s:= |S|$.

Corollary~\ref{cor.logistic} below follows by combining Theorem~\ref{glmnonlin} with existing results on $\ell_1$-penalized logistic regression.  We assume conditions, stated formally in Lemma~\ref{logit.rates} in Appendix \ref{sec:appendix}, which guarantee that with high probability this penalized estimator is sufficiently close to $\beta_0$.

\begin{cor}\label{cor.logistic}
Assume that the conditions of Lemma \ref{logit.rates} in Appendix \ref{sec:appendix} hold
 and in addition assume that $\hat\beta_A\in \Theta(\lambda,\beta_0,X_A)$ and that $12 c_0^{-2}Ks\lambda_A\leq 1$ where $c_0^2 = (e^\eta / \epsilon_0 + 1)^{-2}$ and $K,\eta$ and $\epsilon_0$ are defined in Lemma \ref{logit.rates}. Suppose that $\lambda\asymp\lambda_{\rm{sq}}\asymp \lambda_A\asymp \sqrt{\log (2p)/n}.$
Then there exists a constant $C>0$ such that for any $z\in\mathbb R$, and whenever $Z_A$ is such that $\hat{f}(X) \neq X\hat{\beta}_{\mathrm{sq}}$, 
\begin{eqnarray}
\label{T3.bound}
\bigl|\mathbb{P}\left(T  - \del < z|{Z_A}\right) - \Phi(z)\bigr| 
\leq C\left(
(2p)^{-1} +  \frac{s\{\log (2p)+K\sqrt{\log (2p)}\}}{\sqrt{n}} + \|\hat w_A\|_\infty {s}\log (2p)  
  \right),
\end{eqnarray}
 where
$$\del := \hat w_A^T \bigl\{\mu\bigl(f_0(X)\bigr) - \mu(X\beta_0)\bigr\}.$$

\end{cor}





\section{Empirical results: Logistic regression}
\label{sec:simulations}
In this section we explore the empirical performance of the methods for goodness-of-fit testing and group testing in the setting of logistic regression.
We begin by considering goodness-of-fit testing in low-dimensional settings in Section \ref{sec:emp.lowdim} and in high-dimensional settings in Section \ref{sec:emp.highdim}.  Goodness-of-fit testing on semi-real data is investigated in Section \ref{sec:emp.semireal}, while in Section \ref{sec:emp.grouptest}, we explore group testing in high-dimensional settings.

\subsection{Low-dimensional settings}
\label{sec:emp.lowdim}
While for low-dimensional settings, there are numerous methods available for testing goodness-of-fit as discussed in Section \ref{sec:literature}, 
we show that our test from Algorithm \ref{rptest} may  be advantageous even here.  We compare the performance of the our test (with residual prediction method being a random forest with default tuning parameter choices) against the Hosmer--Lemeshow  $\hat C$ test, the Hosmer--Lemeshow $\hat H$ test (see \cite{lemeshow1982review}) and  the
le Cessie--van Houwelingen--Copas--Hosmer unweighted sum of squares test (see \cite{hosmer1997comparison}). These tests are implemented in the function \texttt{HLgof.test()} in the R package \texttt{MKmisc} \citep{MKmisc}.

\par
We simulated data from a logistic regression model with sample size $N=300$ 
and $p=10$ covariates according to
$$Y_i|x_i=u \sim \text{Bern}\bigl(\pi(u)\bigr),$$
 where $$\pi(u) := \mu(u_1 + u_2 + u_3 + \sigma g(u)).$$
We considered different forms for the misspecification $g(\cdot)$: 
\begin{itemize}
\item
quadratic effect: 
(a) $g(u) = 2u_1^2$, (b) $g(u) = 2u_5^2$,
\item
interactions:  \\
(c) $g(u) =  u_1u_2$, (d) $g(u) =  u_1u_3$, (f) $g(u)=u_1u_4$, (g) $g(u)=u_4u_7$.
\end{itemize}
Here $\sigma \geq 0$ measures the size of departure from the null hypothesis $H_0:\sigma=0$. 
 Note that our GRP testing methodology requires an auxiliary sample of size $n_A$. We therefore randomly split the sample taking $n_A=n=N/2$, with $n$ being the number of observations in the main sample.
The observation vectors $x_i$ follow a $\mathcal N_{10}(0,\Sigma_0)$ distribution where
\begin{equation} \label{eq:Toep}
(\Sigma_0)_{ij} := \rho^{|i-j|}
\end{equation}
is the Toeplitz matrix with correlation $\rho=0.6$.
 The results for the six settings above are shown in Figure \ref{fig:lowdim}. 
All methods maintain good control over type I error, but in most scenarios our GRP-test has significantly greater power compared with the other methods.

\begin{figure}[h!]
\centering \bf
Testing goodness-of-fit of logistic regression: Power comparison \\\vspace{0.5cm}
\begin{subfigure}{0.3\textwidth}
\includegraphics[width=\linewidth]{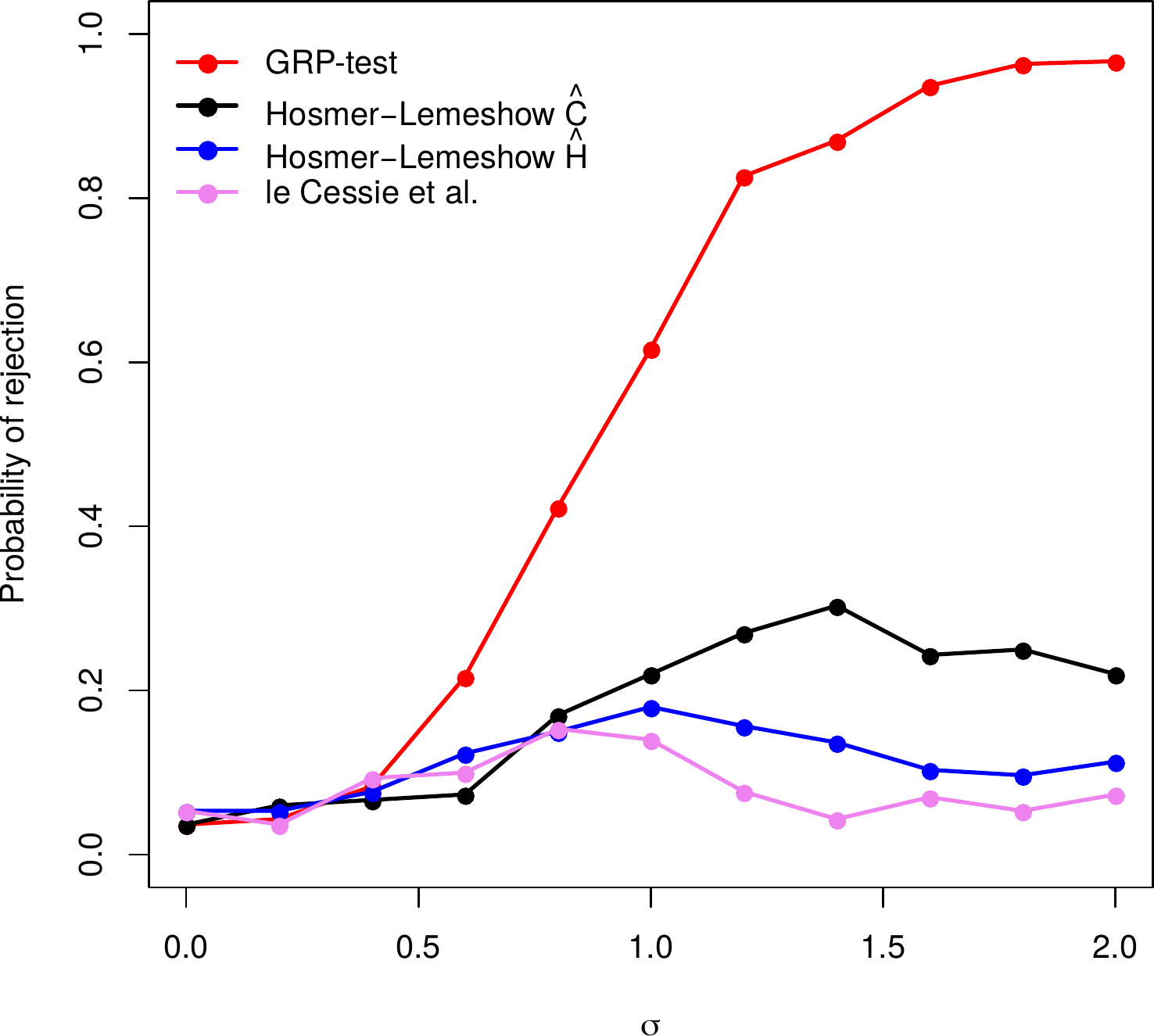}
\caption{Quadratic effect $u_1^2$} \label{fig:1a}
\end{subfigure}
\begin{subfigure}{0.3\textwidth}
\includegraphics[width=\linewidth]{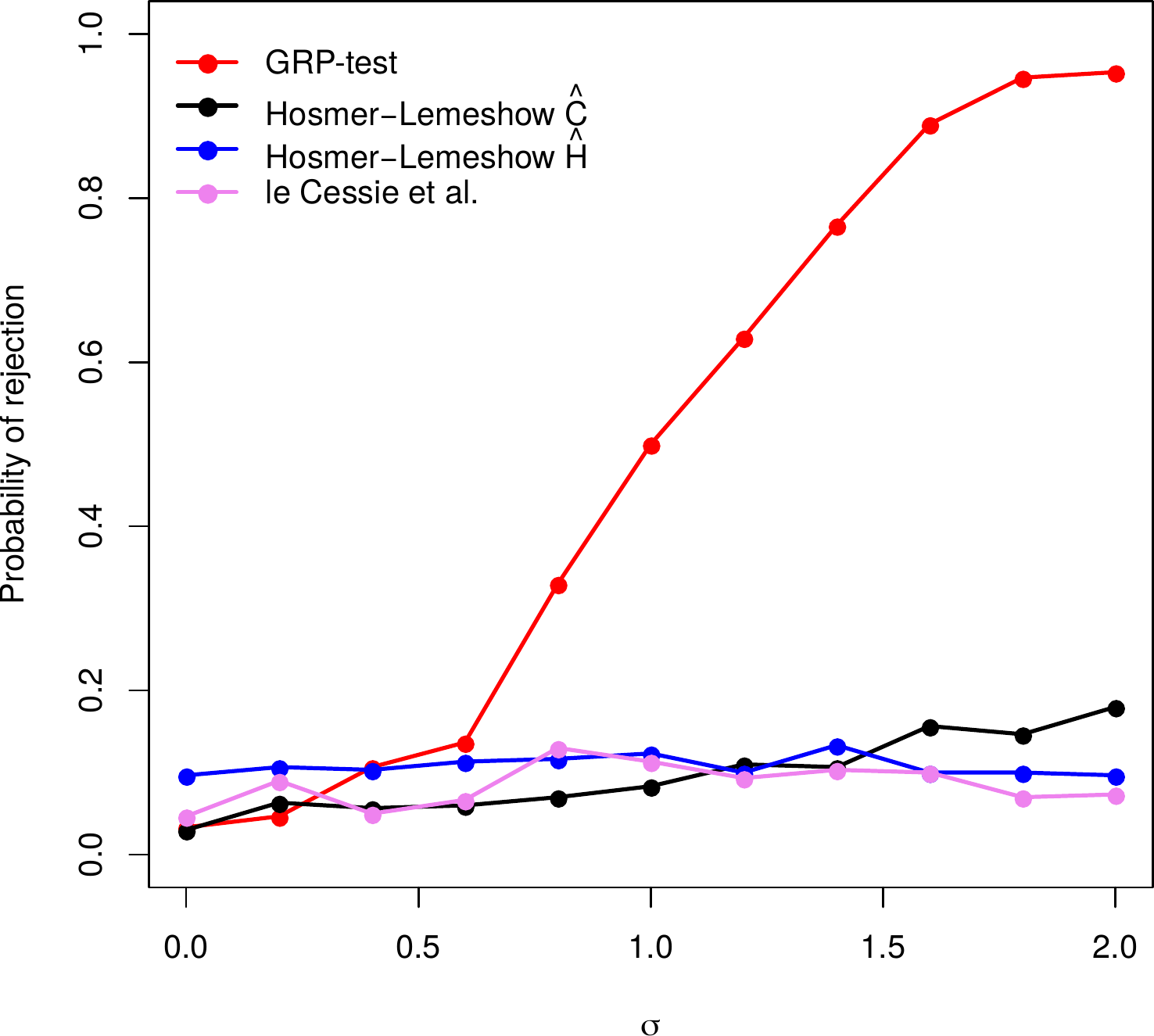}
\caption{Quadratic effect $u_5^2$ 
} \label{fig:1c}
\end{subfigure}
\begin{subfigure}{0.3\textwidth}
\includegraphics[width=\linewidth]{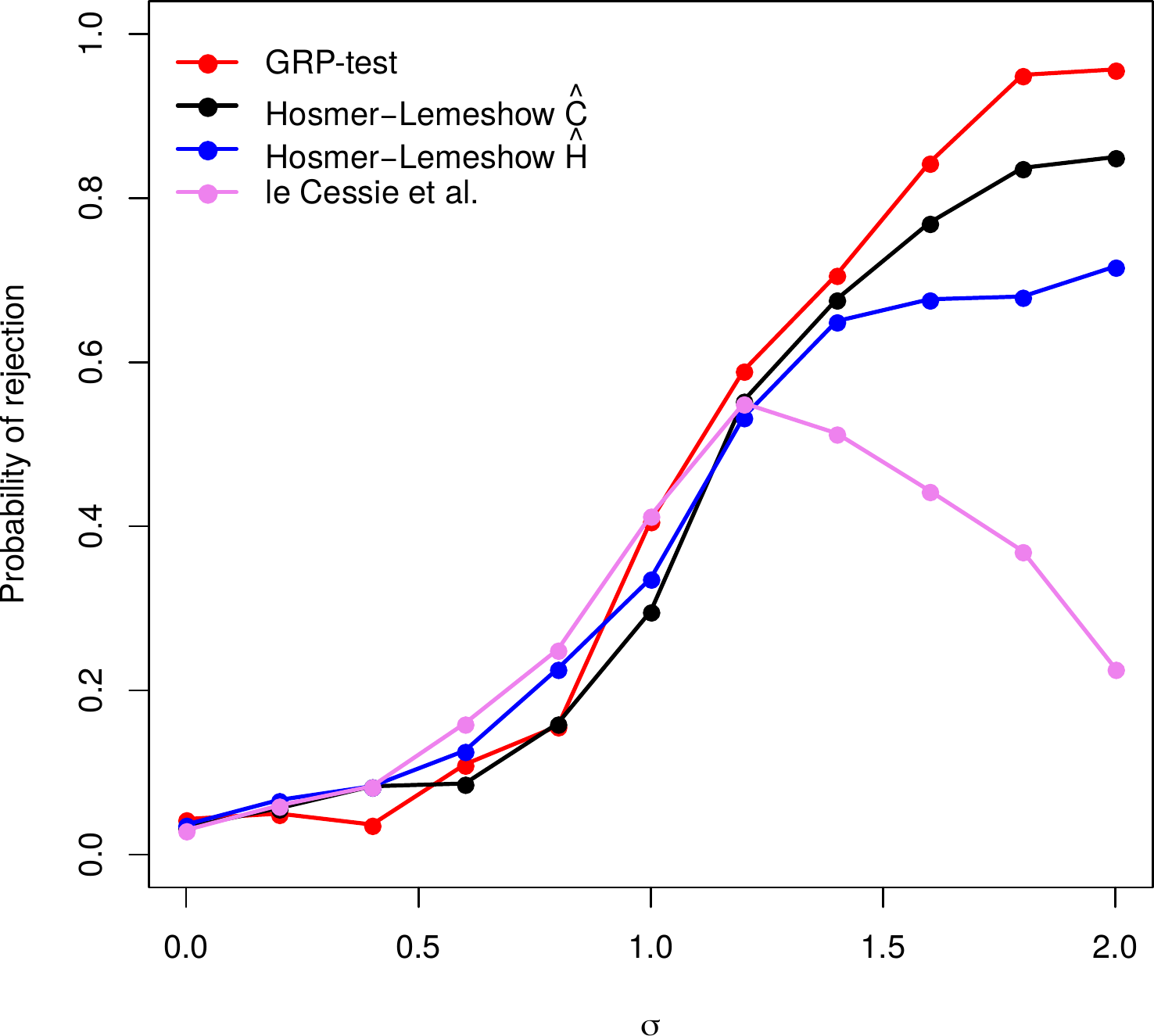}
\caption{Interaction $u_1u_2$} \label{fig:1c}
\end{subfigure}
\\\vspace{0.5cm}
\begin{subfigure}{0.3\textwidth}
\includegraphics[width=\linewidth]{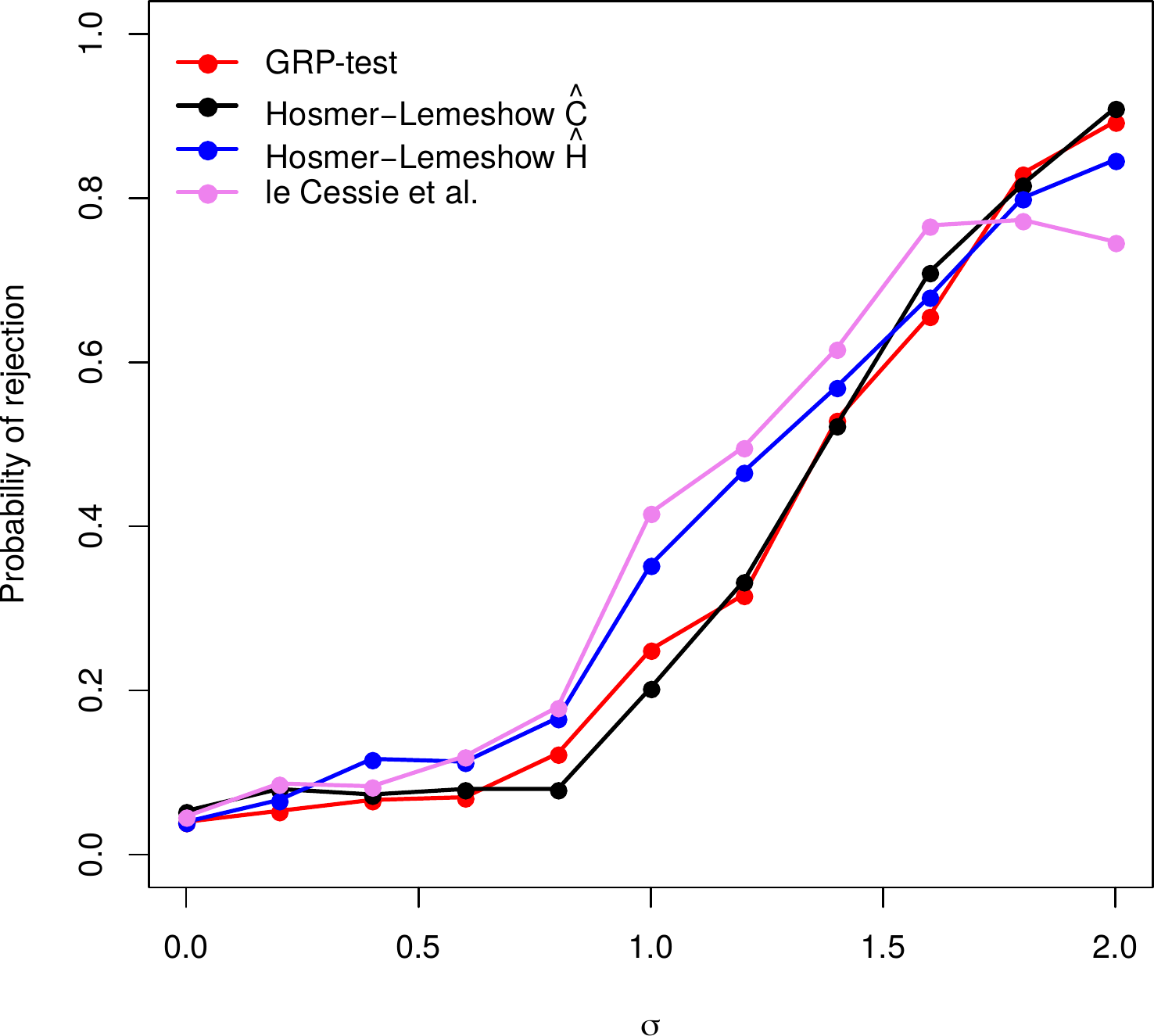}
\caption{Interaction $u_1u_3$} \label{fig:1b}
\end{subfigure}
\begin{subfigure}{0.3\textwidth}
\includegraphics[width=\linewidth]{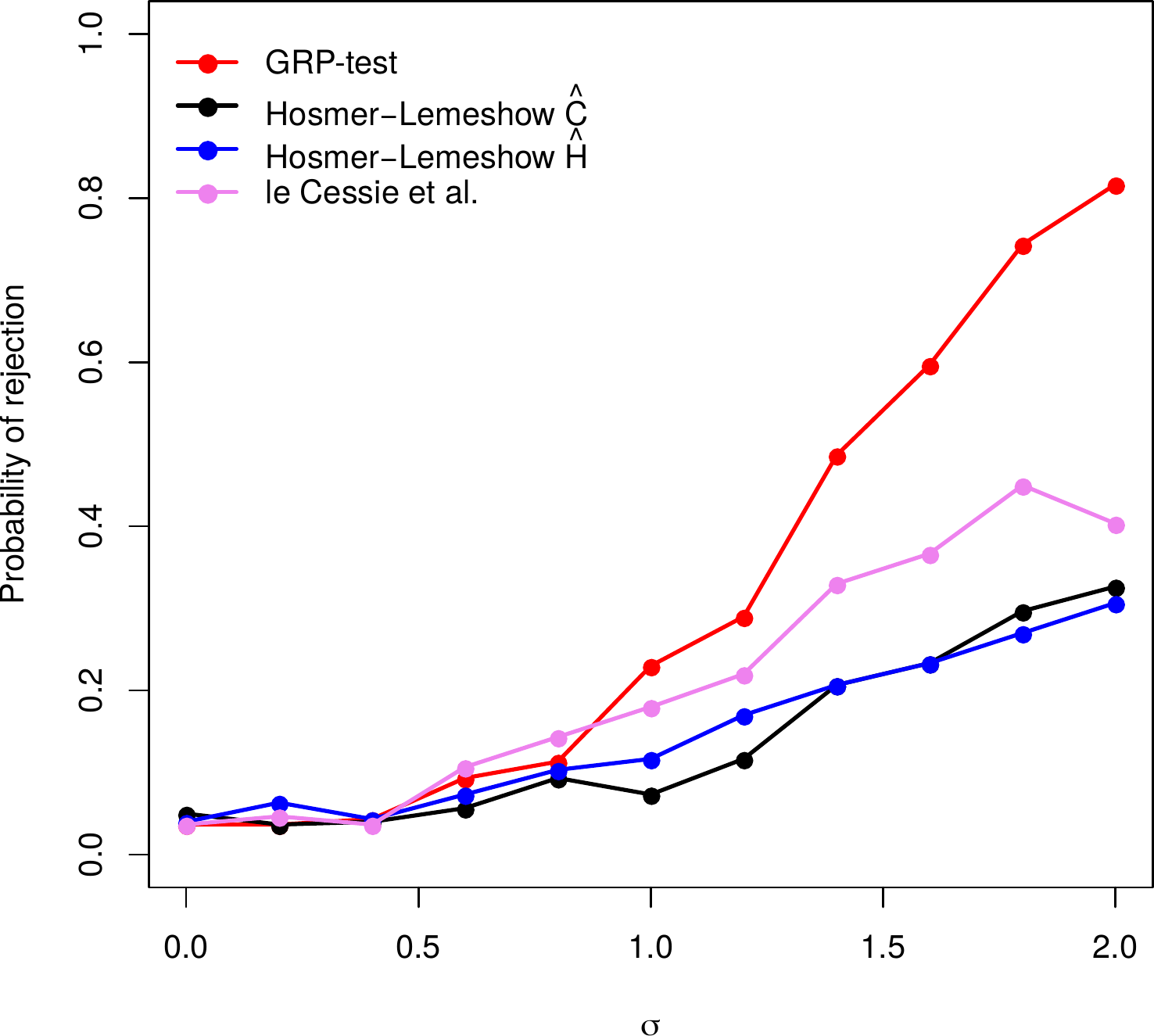}
\caption{Interaction $u_1u_4$} \label{fig:1c}
\end{subfigure}
\begin{subfigure}{0.3\textwidth}
\includegraphics[width=\linewidth]{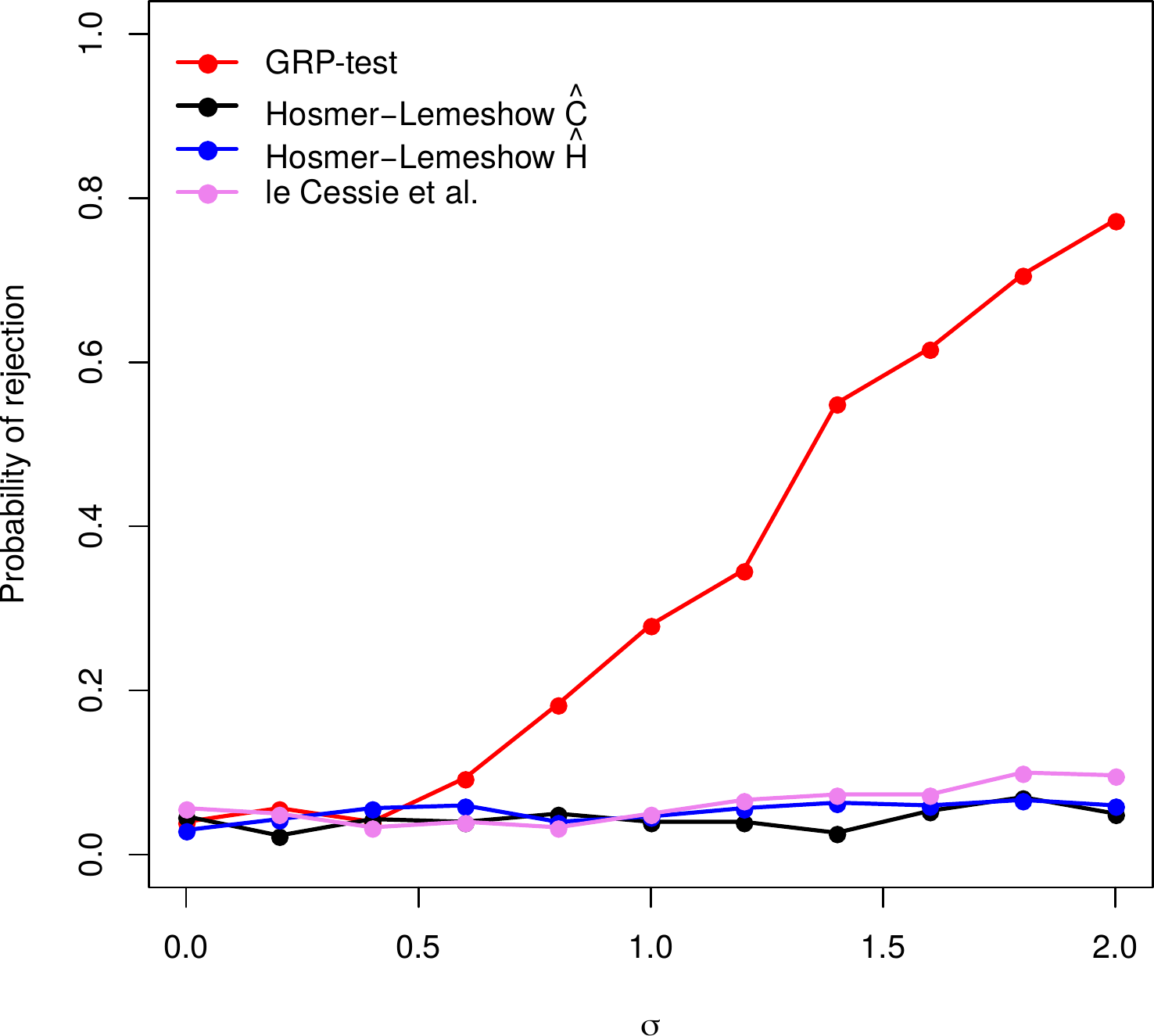}
\caption{Interaction $u_4u_7$} \label{fig:1c}
\end{subfigure}

\caption{Comparison of GRP-test with Hosmer--Lemeshow $\hat C,$ Hosmer--Lemeshow  $\hat H$ and le Cessie--van Houwelingen--Copas--Hosmer unweighted sum of squares test.}
\label{fig:lowdim}
\end{figure}


\subsection{High-dimensional settings}
\label{sec:emp.highdim}
Here we consider logistic regression models with two different types of misspecification from the presence of a pure quadratic effect and an interaction term. 
Specifically, we take the log-odds to be
\[
f_0(u) = u^T\beta_0 + \sigma g(u_1,\dots,u_p)
\]
with
$$\beta_0:=(1,1,1,1,1,0,\dots,0)\in \mathbb R^p$$
and
either (a) $g(u) = u_1u_2$ or (b) $g(u) = u_1^2/2$. The parameter $\sigma$ controls the degree of the misspecification and we look at $\sigma \in [0,3]$. 

The observation vectors $\exi$ are independent with a Gaussian distribution $\mathcal N_p(0,\Sigma_0),$
where $\Sigma_0$ is the Toeplitz matrix {\eqref{eq:Toep}} for a range of correlations $\rho\in \{0.4,0.6,0.8\}$.  
We consider a setting with $p=500$ and $N=800.$
The GRP-test requires sample splitting and we use split size $50\%$; thus the size of auxiliary sample is $n_A = 400$. 
Again, we use a random forest as the residual prediction method. 

In order to achieve better control of the type I error, we use a slight modification of the procedure proposed in Algorithm~1. Define $\hat S :=\{j:\hat\beta_j\neq 0\}$.  Rather than calculating the direction $\hat{w}_A$ through $\hat{\beta}_{\mathrm{sq}}$, we instead use
\[
\tilde{\beta}_{\mathrm{sq}} := \argmin_{\beta \in \R^p} \left\{\frac{1}{\sqrt{n}} \bigl\|\hat{D}_A\bigl(\hat{f}(X)-X\beta\bigr)\bigr\|_2 + \lambda_{\mathrm{sq}}\|\beta_{\hat{S}^c}\|_1 \right\}
\]
in its place within \eqref{eq:w_hat_def}. In this way, we enforce that $\hat{w}_A$ is \emph{exactly} orthogonal to $\hat{D}_A X_{\hat{S}}$. This helps to keep
the remainder term arising from the asymptotic expansion of the test statistic under control, as can  be seen from the following argument. Assume that a ``beta-min condition'' is satisfied, that is 
for all $j\in S := \{k\in \{1,\dots,p\}:\beta_{0,k}\not = 0\}$ it holds that $\beta_{0,j} \gtrsim \sqrt{s\log p/n}$, where $s:=|S|$. 
Then asymptotically, it holds that $\hat{S} \supseteq S$ 
with high probability  \citep[e.g.][Corollary~7.6]{hds}.  
On the event that this occurs, we have that the remainder term in  \eqref{eq:decomp} satisfies
\begin{align*}
\hat{w}_A^TD_{\beta_0}X(\beta_0 - \hat{\beta}) &\approx \hat{w}_A^T\hat{D}_AX(\beta_0 - \hat{\beta}) \\
&= \hat{w}_A^T\hat{D}_AX_{\hat{S}}(\beta_{0,\hat{S}} - \hat{\beta}_{\hat{S}}) =0.
\end{align*}
Even without such a beta-min condition, it is plausible that we will obtain a reduction in this bias term through this strategy of exact orthogonalization.

\par
In the high-dimensional setting, there is no obvious method that we can use for comparison with our proposed GRP-test.
Therefore as a theoretical benchmark, we consider an oracle GRP-test applied to a  reduced design matrix containing only variables in the active set $S=\{1,\dots,5\}$, thereby reducing problem to a low-dimensional one. The results are reported in Tables \ref{tab:m1} and~\ref{tab:m2}, from which we see that the GRP-test does indeed control the type I error, and suffers only a relatively small loss in power compared with the oracle GRP-test.

\begin{table}[h!]
\centering
Detecting the quadratic effect $\sigma u_1^2$\vskip 0.2cm
\begin{tabular}{ccccc}
  \hline
	 {\bf GRP-test } &  $\sigma = 0$ & $\sigma = 0.5 $ &$\sigma =1 $ & $\sigma =1.5$
	\\
	\hline
	 $\rho=0.4$ & 0.02 & 0.16 & 0.86 & 0.96 \\
	 $\rho=0.6$ & 0.04 & 0.18 & 0.82 & 0.94  \\  
	 $\rho=0.8$ & 0.06 & 0.12 & 0.52 & 0.96 \\[0.5cm] 

	\hline
	{\bf Benchmark} &$\sigma = 0$ & $\sigma = 0.5 $ &$\sigma =1 $ & $\sigma =1.5$ \\\hline
	 $\rho=0.4$ & 0.05 & 0.52 & 0.99 & 1.00\\
	 $\rho=0.6$ & 0.02 & 0.35 & 0.92 & 1.00\\%
	 $\rho=0.8$ & 0.05 & 0.18 & 0.76 & 0.99 
   \\\hline
\end{tabular}
\caption[]{
Estimated probabilities of rejection of $H_0:\sigma=0$ at significance level 0.05 for different values of $\rho$ and $\sigma$. 
Dimensions of the data are $p=500$ and $N=800$. 
Averages for the GRP-test were calculated over 50 iterations. 
}
\label{tab:m1}

\end{table}

\begin{table}[h!]
\centering
Detecting the interaction effect $\sigma u_1u_2$\vskip 0.2cm
\begin{tabular}{ccccc}
  \hline
 {\bf GRP-test } &  $\sigma = 0$ & $\sigma = 1 $ &$\sigma =2 $ & $\sigma =3$
	\\ 
  \hline
	 $\rho=0.4$ & 0.02 & 0.16 & 0.86 & 0.94   \\
	 $\rho=0.6$ & 0.04 & 0.14 & 0.96 & 1.00    \\
	 $\rho=0.8$ & 0.06 & 0.34 & 1.00 &  1.00   \\[0.5cm]
	\hline
	{\bf Benchmark} &$\sigma = 0$ & $\sigma = 1 $ &$\sigma =2 $ & $\sigma =3$ \\\hline
	 $\rho=0.4$ & 0.05 & 0.68 & 1.00 & 1.00 \\ 
	 $\rho=0.6$ & 0.04 & 0.70 & 1.00 & 1.00\\ 
	 $\rho=0.8$ & 0.04 & 0.38 & 1.00 & 1.00\\
   \hline
\end{tabular}
\caption[]{
Estimated probabilities of rejection of $H_0:\sigma=0$ at significance level 0.05 for different values of $\rho$ and $\sigma$. 
Dimensions of the data are $p =500$ and $ N=800.$
}
\label{tab:m2}

\end{table}

\subsection{Semi-real data example}
\label{sec:emp.semireal}
We use a gene-expression dataset on lung cancer available from the NCBI database 
(\cite{gene.data}, \texttt{https://www.ncbi.nlm.nih.gov/sites/GDSbrowser?acc=GDS2771}) to illustrate the size and power performance of the goodness-of-fit test. We aim to detect if the model is a logistic regression, or if there are extra nonlinear effects. 
The full dataset contains airway epithelial gene expressions for 22215 genes from each of 192 smokers with (suspected) lung cancer, but this was reduced by taking the 500 genes with the largest variances.  Having scaled the resulting variables, we fit a $\ell_1$-penalized logistic regression using \texttt{cv.glmnet()} from the package \texttt{glmnet} \citep{glmnet} and obtained a parameter estimate $\hat\beta$ with its corresponding support set $\hat S$.  We then fit a Gaussian copula model to the rows of the design matrix and generated a new, augmented design matrix $X$ by simulating a further $N=800$ observation vectors from this fitted model. 
Finally, we generated 800 new responses: $Y_i |x_i=u \sim \text{Bern}\bigl(\hat\pi(u)\bigr)$
where 
$$ \hat\pi(u) := u^T\hat\beta + 3\frac{g(u) - \bar{g}}{\hat{\sigma}_g}\mathds{1}_{\{\hat{\sigma}_g \neq 0\}},$$ 
$\bar{g} := N^{-1}\sum_{i=1}^N g(x_i)$ and $\hat{\sigma}_g^2 := (N-1)^{-1}\sum_{i=1}^N \{g(x_i) - \bar{g}\}^2$, for  the following three scenarios:
\begin{align*}
g(u) &= 0,
\\
g(u) &= u_{j_1}^2 + u_{j_2}^2,
\\
g(u) &= u_{j_1} u_{j_2} + u_{j_3} u_{j_4},
\end{align*}
where $j_{\ell} \in \hat S$, $\ell=1,\dots,4$ are uniformly sampled entries from $\hat S.$
We report rejection probabilities for all three scenarios from 100 repetitions in Table \ref{tab:semireal}.  In each case, the GRP-test is able to detect the misspecification relatively reliably, while keeping the type I error under control.

\begin{table}[h!]
\centering
\textbf{Testing goodness-of-fit of logistic regression on semi-real data on lung cancer}
\vskip 0.2cm
\begin{tabular}{lccc}
  \hline
 {\bf  } &  Prob. of rejection of $H_0$ 
	\\ 
  \hline
	$g(u) = 0$ & 0.05 
	\\
	$g (u ) = u_{j_1}^2 + u_{j_2}^2$ & 0.77 
	\\
	$g(u) = u_{j_1} u_{j_2} + u_{j_3} u_{j_4}$ &  0.81 
	\\
   \hline
\end{tabular}
\caption[]{
Estimated probabilities of rejection of $H_0:g(u)=0$ at significance level 0.05, averaged over 100 generated datasets.
}
\label{tab:semireal}

\end{table}

\begin{figure}[h!]
\centering \bf
Testing goodness-of-fit of logistic regression on semi-real data on lung cancer \\
\begin{subfigure}{0.3\textwidth}
\includegraphics[width=\linewidth]{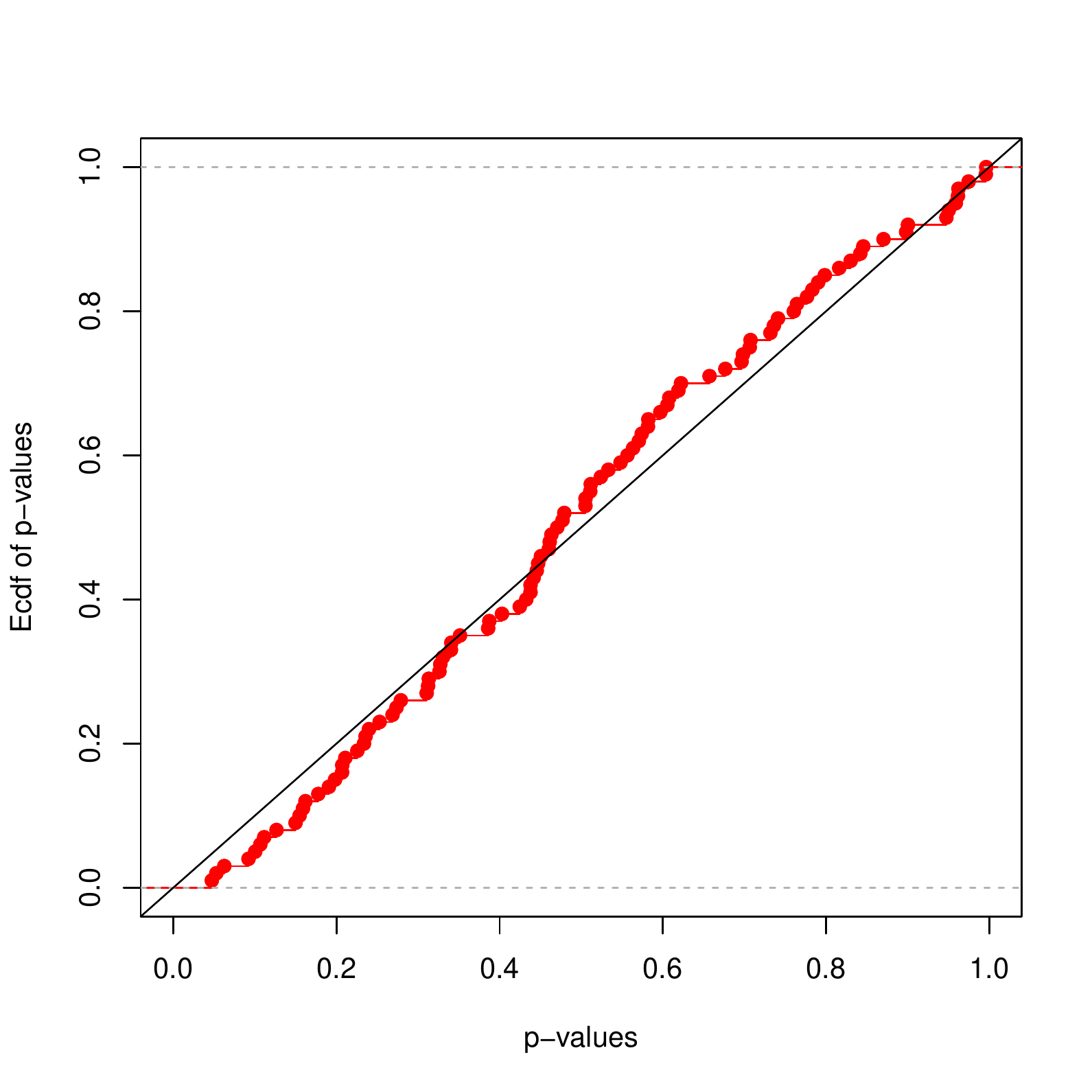}
\caption{No misspecification: $g(u) = 0$.} \label{fig:1a}
\end{subfigure}
\begin{subfigure}{0.3\textwidth}
\includegraphics[width=\linewidth]{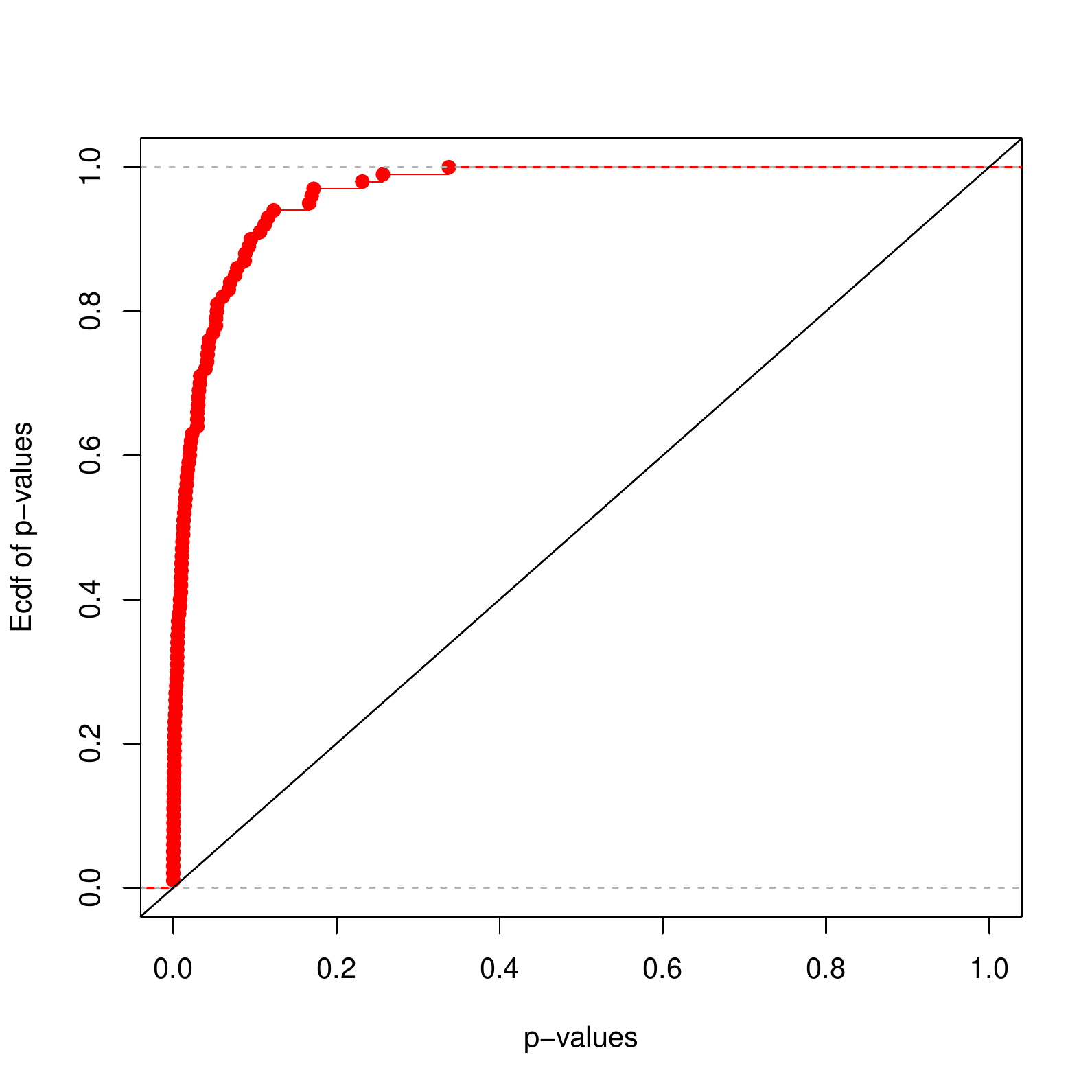}
\caption{Quadratic effects: $g(u) = u_{j_1}^2 + u_{j_2}^2.$ 
} \label{fig:1c}
\end{subfigure}
\begin{subfigure}{0.3\textwidth}
\includegraphics[width=\linewidth]{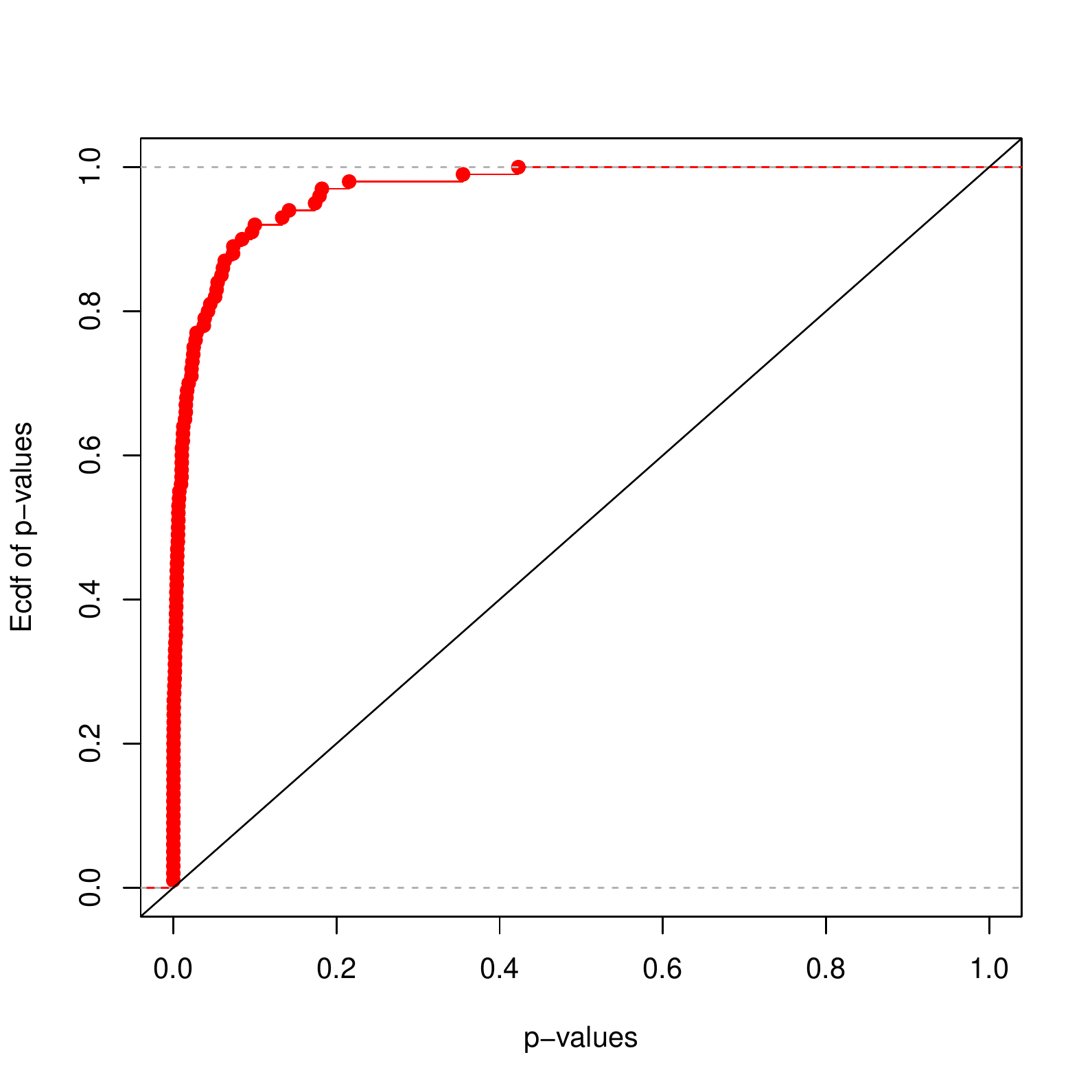}
\caption{Interactions: $g(u) = u_{j_1}u_{j_2} + u_{j_3}u_{j_4}$.} \label{fig:1c}
\end{subfigure}

\caption{Empirical distribution functions of p-values for testing the goodness-of-fit in the three scenarios.}
\label{fig:lowdim}
\end{figure}


\subsection{Group testing}
\label{sec:emp.grouptest}
Finally, we consider the problem of testing for the significance of groups of predictors using the methodology set out in Section~\ref{subsec:group.testing}. We compare the GRP-test (Algorithm~\ref{grptest}) with the globaltest \citep{goeman2004global} and the de-biased Lasso \citep{vdgeer13, dezeure2015high} for logistic regression. 
We consider logistic regression models with coefficient vector of the form
\[
\beta_0 = (1, 1, 1, 1, \theta, 0, \ldots, 0) \in \R^p
\]
for a range of values of $\theta \in [0, 1.5]$, and look at testing the null hypothesis $\beta_{0,G}=0$ where $G=\{5, 6, \ldots, p\}$. Thus larger values of $\theta$ correspond to more extreme violations of the null. Similarly to earlier examples, we use design matrices constructed via realisations of a random Gaussian design with Toeplitz covariance \eqref{eq:Toep} and $\rho=0.6$.
The results are reported in Figures \ref{fig:group1} and \ref{fig:group2}.  Both the GRP-test and the globaltest control the type I error at very close to the nominal level, while from Figure~\ref{fig:group1} the de-biased Lasso test is conservative; on the other hand, the GRP-test does very well in these examples in terms of power.  

\begin{figure}[h!]
\centering
\textbf{Group testing in logistic regression: comparison of GRP-test, de-biased Lasso and globaltest
} \\
\vspace{0.5cm}
\vskip 0.0cm
\begin{subfigure}{0.3\textwidth}
\includegraphics[width=\linewidth]{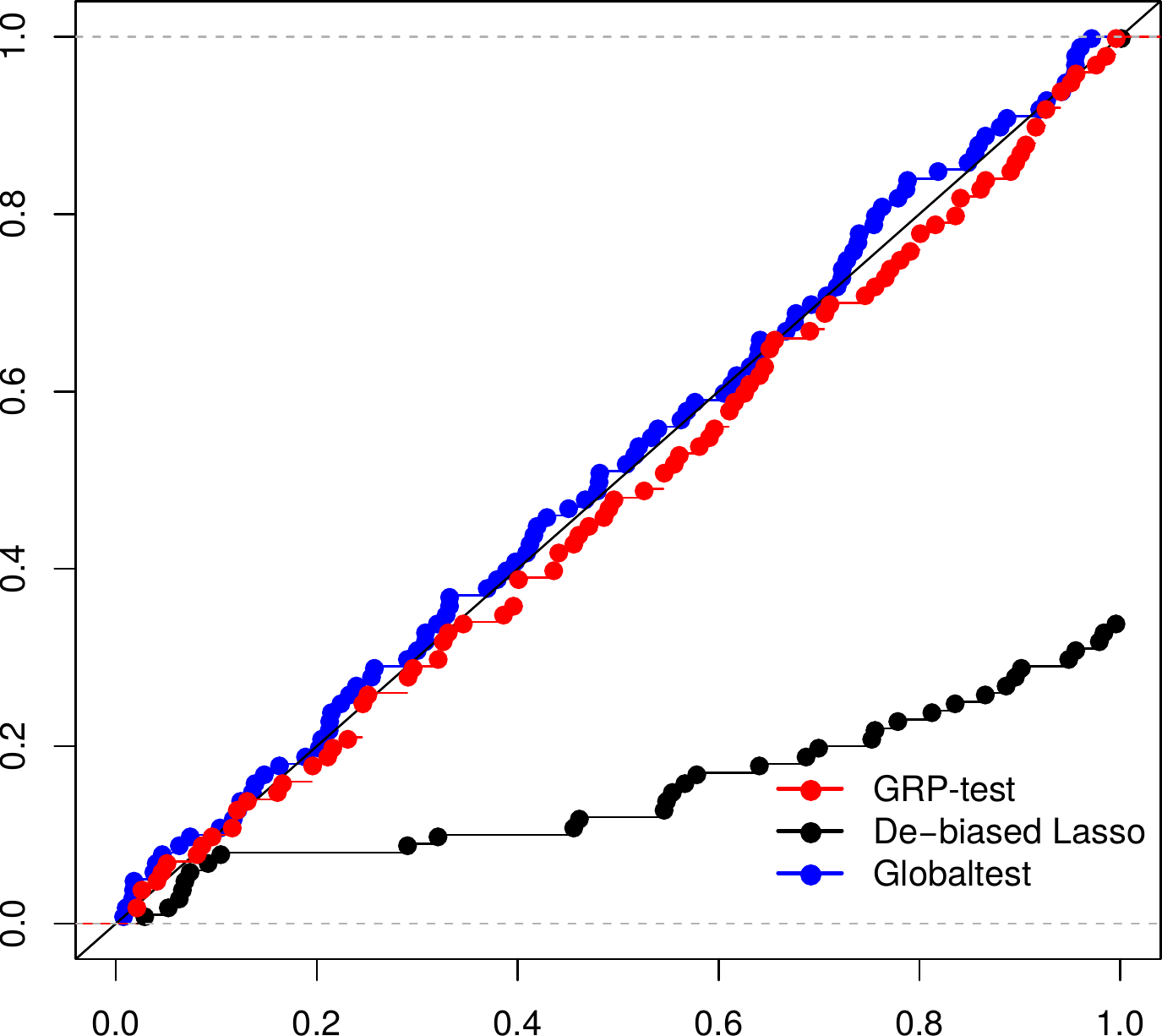}
\caption{ $\theta=0$ ($H_0$ true)} \label{fig:1a}
\end{subfigure}
\hspace*{0.5cm} 
\begin{subfigure}{0.3\textwidth}
\includegraphics[width=\linewidth]{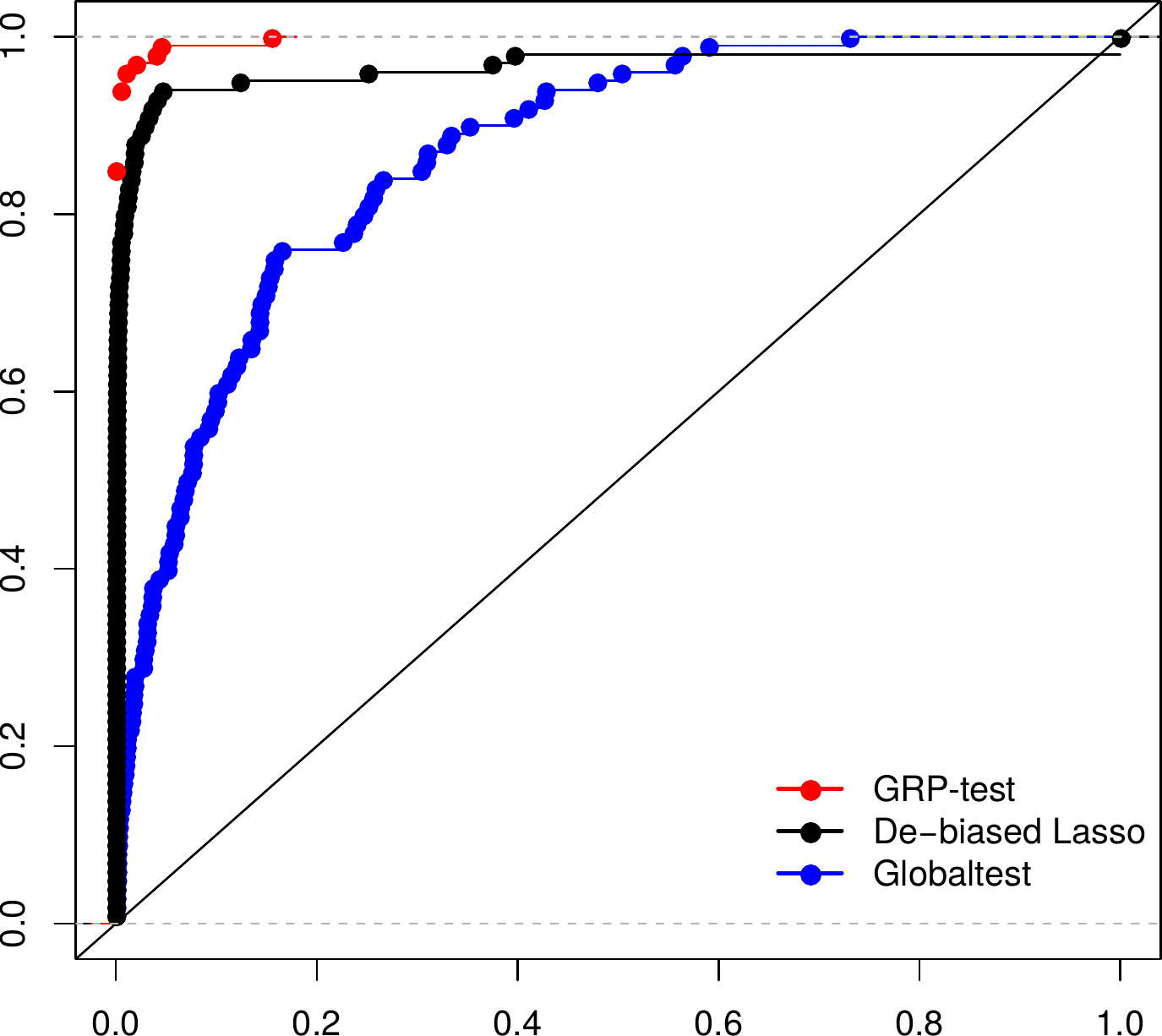}
\caption{$\theta=1$ ($H_0$ false)} \label{fig:1b}
\end{subfigure}
\caption{
Comparison of GRP-test from Algorithm \ref{grptest} with the de-biased Lasso \citep{dezeure2015high} and globaltest \citep{goeman2004global} when testing $H_0: \beta_{0,G} = 0$, with $G = \{5,\dots,p\}$, where $ \beta_0=(1,1,1,1,\theta,0,\dots,0).$ Plots show the empirical distribution functions of p-values under the null hypothesis (left) and under the alternative $\theta=1$ (right). The dimensions of the data are $n=500, p = 100$.
} 
\label{fig:group1}
\end{figure}

\begin{figure}[h!]
\centering
\textbf{Group testing in logistic regression:  comparison of GRP-test and globaltest }\\
\vskip 0.0cm
\vspace{0.5cm}
\begin{subfigure}{0.3\textwidth}
\includegraphics[width=\linewidth]{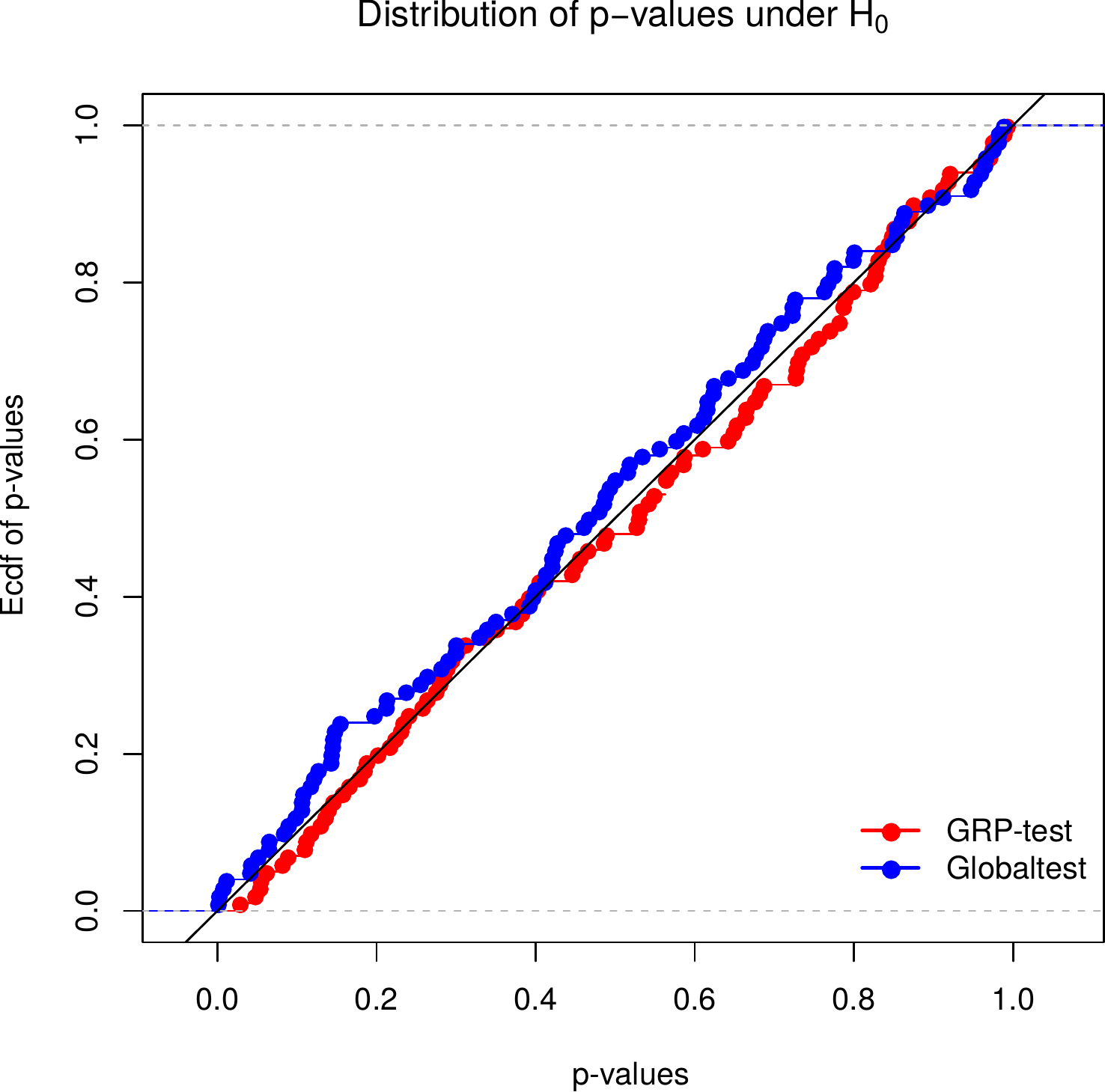}
\caption{ $\theta = 0$ ($H_0$ true)} \label{fig:1a}
\end{subfigure}
\hspace*{0.5cm} 
\begin{subfigure}{0.3\textwidth}
\includegraphics[width=\linewidth]{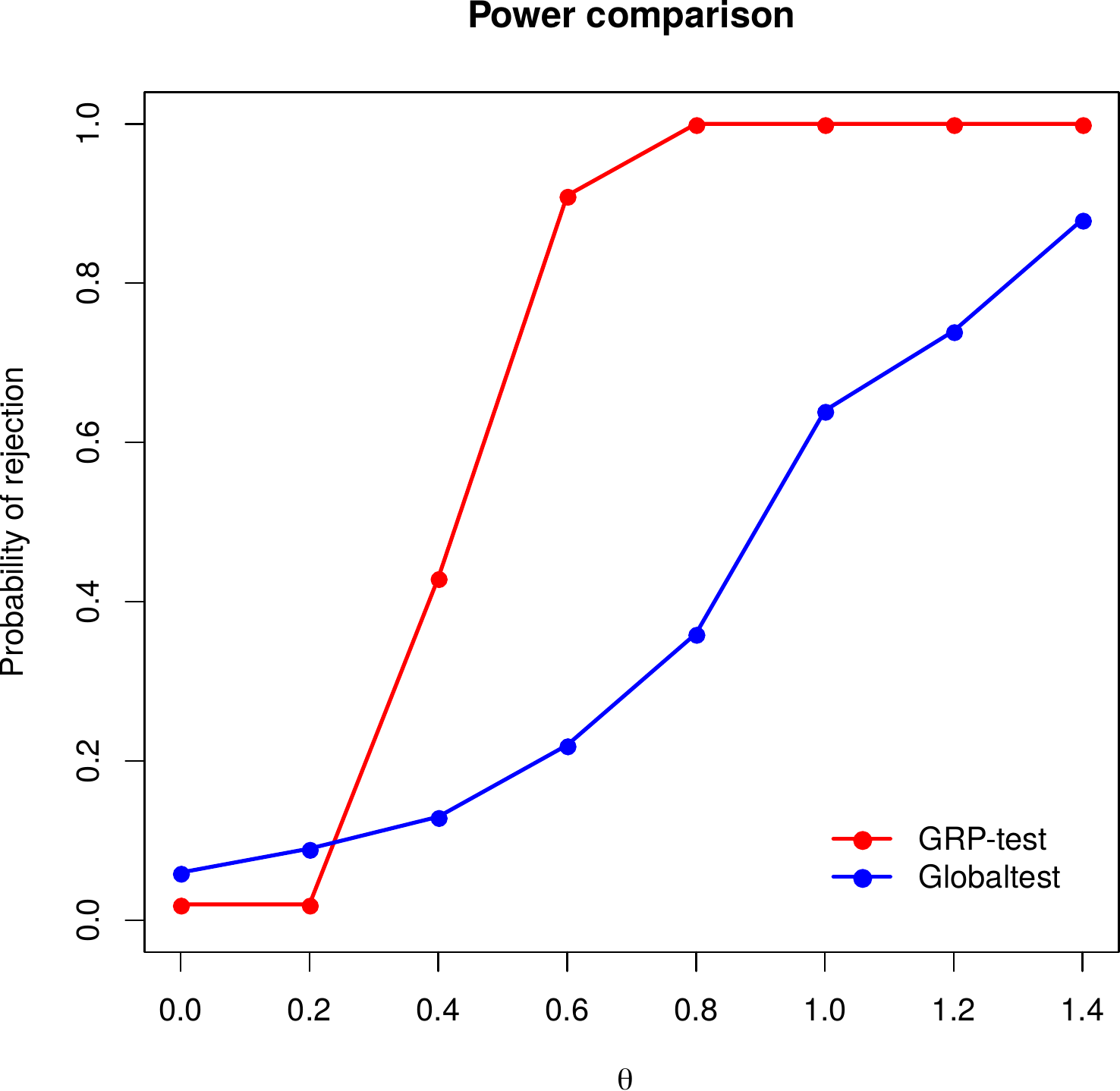}
\caption{Powers for a range of $\theta$ ($H_0$ false)} \label{fig:1b}
\end{subfigure}
\caption{
Comparison of GRP-test from Algorithm \ref{grptest} and globaltest \citep{goeman2004global}
 when testing $H_0: \beta_{0,G} = 0,$ with $G = \{5,\dots,100\}$, where
$ \beta_0=(1,1,1,1,\theta,0,\dots,0).$
 Plots show the empirical distribution functions of p-values under the null hypothesis (left) 
and powers of the test at significance level $\alpha = 0.05$ for a range of values of $\theta$ which are on the $x$-axis (right). 
The dimensions of the data are $p=800, n =600$.  The de-biased Lasso test is omitted since it is very computationally expensive.
} 
\label{fig:group2}
\end{figure}

\section{Discussion}
\label{sec:discuss}
In this work, we have introduced a new method for detecting conditional mean misspecification in generalized linear models based on predicting remaining signal in the residuals. For this task of prediction, we have a number of powerful machine learning methods at our disposal. Whilst these estimation performance of these methods is largely theoretically intractable, by employing sample-splitting and a careful debiasing strategy involving the square-root Lasso, our generalized residual prediction framework provides formal statistical tests with type I error control when used in conjunction with (essentially) arbitrary machine learning methods.

One requirement for these theoretical guarantees is that the sparsity of the true regression coefficient $\beta_0$ satisfies $s=o(\sqrt{n}/\log p) $, a condition that was also needed in related work on the de-biased Lasso \citep{zhang,vdgeer13,stanford1}. It would be very interesting if this could be relaxed to $s = o(n / \log p)$ for instance, which would encompass settings where the GLM Lasso estimate satisfies $\|\hat{\beta} - \beta_0\|_2 \to 0$ though $\|\hat{\beta} - \beta_0\|_1$ may be diverging.

Another interesting question is whether sample splitting can be completely avoided if we were able to obtain guarantees for an estimator $\hat w$ of a population direction $w$. Such alternatives to sample splitting could be particularly helpful for settings where there is dependence across the observations, such as in the case of generalized linear mixed effect models.

\vspace{0.3cm}

\textbf{Acknowledgements}: The authors would like to thank the Isaac Newton Institute for Mathematical Sciences for support and hospitality during the `Statistical Scalability' programme when work on this paper was undertaken, supported by EPSRC grant number EP/R014604/1.  JJ is supported by a Swiss National Science Foundation fellowship.  RDS is supported by an EPSRC First Grant and an EPSRC programme grant.  PB is supported by the European Research Council under the grant agreement No. 786461 (CausalStats -- ERC-2017-ADG).  RJS is supported by an EPSRC fellowship and an EPSRC programme grant.

\bibliography{gminf}

\input{proofs}

\end{document}

%% file: proofs.tex
\section{Proofs}
\label{sec:proofs}

%

\subsection{Proofs for Section \ref{sec:theory}}
\subsubsection{Proofs for Section \ref{subsec:rptest}}

\begin{proof}[Proof of Theorem \ref{glmnonlinH0}]
This follows from the more general Theorem \ref{glmnonlin}, noting that under the null hypothesis, $\del = 0$.
The proof of Theorem \ref{glmnonlin} can be found in Section \ref{sec:power.proof}.
\end{proof}

\subsubsection{Proofs for Section \ref{subsec:group.testing}}

\begin{proof}[Proof of Proposition \ref{grouptest.prop}]
For $j \in G$, let $u_j := X_j - X_{-G}\gamma_{0,j}$ and define the sets 
$$\mathcal T_{1,j} := \left\{\|u_j^TD_{\beta_0} X_{-G}\|_\infty/n \leq 6 C_0 K{^2}\sqrt{\log (2p)/n}\right\},$$
$$\mathcal T_{2}   := \left\{ \hat\beta_{-G}\in\Theta_{-G}(\lambda,\beta_0)\right\},$$ 
$$\mathcal T_{3,j} := \left\{\| X_{j\cup G^c}^T (Y - \mu(X_{}\beta_{0}))\|_\infty/{n} \leq AK\sqrt{\log p/n}\right\},$$
where $C$ will be specified below.
We first derive a high-probability bound for the set 
$$\mathcal T:= \cap_{j\in G} (\mathcal T_{1,j} \cap \mathcal T_2 \cap \mathcal T_{3,j}).$$ 
By Lemma \ref{prob.bounds} in Section \ref{sec:appendix}, there exists a constant $A>0$ (in the definition of $\mathcal T_{3,j}$) such that 
\begin{equation}
  \label{eq:Tc}
  \mathbb{P}(\mathcal T^c) \leq 3/p.
\end{equation}
In the rest of the proof, we work on the event $\mathcal T.$ 
Define $D_{\beta_{0,-G}} := \mu'(X_{-G}\beta_{0,-G})$ (note that under $H_0$, it holds that $D_{\beta_{0,-G}} =D_{\beta_0}$). 
Now consider the decomposition
\begin{align}
\label{dTj}
T_j = {\hatwj ^T D_{\hat\beta_{-G}}^{-1} \bigl(Y - \mu(X_{-G}\hat\beta_{-G})\bigr)}
&=
  \wj ^T  D_{\beta_{0,-G}}^{-1} \bigl(Y - \mu(X_{-G}\beta_{0,-G})\bigr)
+\text{rem}_{1,j} 
                                                                            +\text{rem}_{2,j} \nonumber \\
&= \wj^T \varepsilon +\text{rem}_{1,j} +\text{rem}_{2,j},  
\end{align}
where
\begin{align*}
\text{rem}_{1,j} := { (  D_{\hat\beta_{-G}}^{-1} \hatwj  -  D_{\beta_{0,-G}}^{-1}\wj ) ^T  \bigl(Y - \mu(X_{-G}\beta_{0,-G})\bigr)}, \\
\text{rem}_{2,j} := { \hatwj ^T  D_{\hat\beta_{-G}}^{-1} \bigl(\mu(X_{-G}\beta_{0,-G}) - \mu(X_{-G}\hat\beta_{-G})\bigr)},
\end{align*}
and where we write $D_{\hat \beta_{-G}}:= \hat D$.   We first derive the rates of convergence for the estimator $\hat\gamma_j$ from Algorithm \ref{grptest} and then
proceed to bound the remainders. 
By Lemma \ref{rates.nodewise} in Section \ref{sec:appendix}, there exist positive constants $C_3$ and $C_4$ such that on 
$\mathcal T,$
we have
\begin{align}
\label{rate.gamma0}
\max_{j\in G}\|\hat\gamma_j - \gamma_{0,j}\|_1 &\leq C_3 K^2s\sqrt{\log p /n},\\
\label{rate.tau0}
\max_{j\in G}|\hat\tau_j^2 - \tau_j^2| &\leq C_4 K^2 \sqrt{s \log p/n},
\end{align}
where $\hat\tau_j^2 := \bigl\| D_{\hat\beta_{-G}}(X_j - X_{-G}\hat\gamma_j)\bigr\|_2^2/n$.

\vskip 0.3cm
\noindent
\textbf{Remainder $\text{rem}_{1,j}$:}
Let $k_j$ denote the index of column $X_j$ in the matrix $X_{j\cup G^c}$.
For $j \in G$, we define 
$$\hat\Gamma_j := (-\hat\gamma_{j,1},\dots,-\hat\gamma_{j,k_j-1},1,-\hat\gamma_{j,k_j+1},\dots,\hat\gamma_{j,|G^c|})^T$$
and its population-level counterpart $\Gamma_{0,j}$ based on $\gamma_{0,j}$.
First note that
$$D_{\hat\beta_{-G}}^{-1} \hatwj = (X_{j}-X_{-G}\hat\gamma_j )/(\sqrt{n}\hat\tau_j) =  X_{j\cup G^c} \hat\Gamma_{j}/(\sqrt{n}\hat\tau_j),$$
and similarly,
\begin{equation}
\label{Eq:Dbeta0}
D_{\beta_{0,-G}}^{-1} \wj = (X_{j}-X_{-G}\gamma_{0,j} )/(\sqrt{n}\tau_j)
=  X_{j\cup G^c}  \Gamma_{0,j}/(\sqrt{n}\tau_j).
\end{equation}
Therefore, we obtain using H\"older's inequality that
\begin{align*} 
|\text{rem}_{1,j}| 
&=
\bigl|{ (    \hat\Gamma_j /\hat\tau_j
- 
   \Gamma_{0,j} /\tau_j ) ^T 
	X_{j\cup G^c}^T (Y - \mu(X_{}\beta_{0}))}\bigr|/\sqrt{n}
\\
&\leq 
\|{     \hat\Gamma_j /\hat\tau_j
- 
   \Gamma_{0,j} /\tau_j \|_1
	\bigl\| X_{j\cup G^c}^T (Y - \mu(X_{}\beta_{0}))}\bigr\|_\infty/\sqrt{n}.
\end{align*}
On the set $\cap_{j\in G} \mathcal T_{3,j}\subseteq \mathcal T,$ we have
	\begin{eqnarray}
	\label{r1j}
\max_{j\in G}|\text{rem}_{1,j}|  \leq \max_{j\in G} \|  \hat\Gamma_j /\hat\tau_j -   \Gamma_{0,j} /\tau_j \|_1AK\sqrt{\log p}.
	\end{eqnarray}
Next we bound 
$\|    \hat\Gamma_j /\hat\tau_j
- 
   \Gamma_{0,j} /\tau_j \|_1.$
Firstly, we can decompose  and bound
	\begin{align}
\nonumber
\|    \hat\Gamma_j /\hat\tau_j
- 
   \Gamma_{0,j} /\tau_j \|_1 
	&= 
	\|    (\hat\Gamma_j -\Gamma_{0,j})/\hat\tau_j  + \Gamma_{0,j} (1/\hat\tau_j -1/\tau_j)  \|_1
	\\
	\label{gamma.decomp}
	&\leq 
	\|\hat\gamma_j -\gamma_{0,j}\|_1/\hat\tau_j + \|\Gamma_{0,j}\|_1|1/\hat\tau_j -1/\tau_j|.
	\end{align}
	Now note that
		\begin{align}
		\nonumber
\tau_j^2 &=
 \mathbb E(X_j - X_{-G}\gamma_{0,j})^T D_{\beta_0}^2(X_j - X_{-G}\gamma_{0,j})/n
\\\nonumber
&=
	\mathbb EX_j ^TD_{\beta_0}^2 (X_j - X_{-G}\gamma_{0,j})/n
	\\\label{deq}
	&=
	 \mathbb E X_j^T D_{\beta_0}^2  X_{j\cup G^c} \Gamma_{0,j}/n.
		\end{align}
	Combining \eqref{deq} with the fact that 
	$$\mathbb EX_{-G}^TD_{\beta_0}^2  X_{j\cup G^c} \Gamma_{0,j}/n=0,$$
we obtain		that
$$\mathbb EX_{j\cup G^c}^TD_{\beta_0}^2  X_{j\cup G^c} \Gamma_{0,j}/n=\tau_j^2e_{k_j},$$
where $e_\ell$ denotes the $\ell$-th standard basis vector in $\mathbb{R}^{p - |G|+1}$. 
Define $\Sigma_{\beta_0,j\cup G^c} := \mathbb E X_{j\cup G^c}^T D_{\beta_0}^2 X_{j\cup G^c}/n$ and note from Condition \ref{grptest.cond}(iv) that $\Sigma_{\beta_0,j\cup G^c}$ is invertible. Consequently,
$\Gamma_{0,j}/\tau_j^2 = \Sigma_{\beta_0,j\cup G^c}^{-1} e_{k_j}$, so $1/\tau_j^2 = (\Sigma_{\beta_0,j\cup G^c}^{-1})_{k_j k_j}.$ 
Further note that  
$$\max_{j\in G}\tau_j^2 =\max_{j\in G} 1/(\Sigma_{\beta_0,j\cup G^c}^{-1})_{k_jk_j} \leq \max_{j\in G}(\Sigma_{\beta_0,j\cup G^c})_{k_jk_j}\leq \|\Sigma_0\|_\infty\leq C_e.$$
Therefore, by Condition \ref{grptest.cond}(iv), we have 
$$\max_{j\in G}\|\Gamma_{0,j}\|_2 
=
\max_{j\in G}\|\Sigma_{\beta_0,j \cup G^c}^{-1} e_j\|_2 \tau_j^2
\leq 
\max_{j\in G} \Lambda^{-1}_{\min}(\Sigma_{\beta_0,j \cup G^c})\tau_j^2
\leq 
\Lambda^{-1}_{\min}(\Sigma_0)\max_{j\in G} \tau_j^2
\leq 
 C_e^2.$$ 
Consequently, and by sparsity of $\gamma_{0,j}$ assumed in Condition \ref{grptest.cond}(v), it follows that
$$\max_{j\in G}\|\Gamma_{0,j}\|_1\leq \sqrt{s+1} \max_{j \in G}\|\Gamma_{0,j}\|_2\leq C_e^2\sqrt{s+1}.$$
Moreover,
\begin{align}
\label{Eq:taulower}
  \min_{j\in G}\tau_j^2 &= \min_{j\in G} \frac{1}{\bigl(\Sigma_{\beta_0,j\cup G^c}^{-1}\bigr)_{k_jk_j}} \geq \min_{j\in G} \frac{1}{\bigl\|\Sigma_{\beta_0,j\cup G^c}^{-1}\bigr\|_{\mathrm{op}}} = \min_{j\in G} \bigl\|\Sigma_{\beta_0,j\cup G^c}\bigr\|_{\mathrm{op}} \nonumber \\
  &\geq 
\min_{j\in G} \Lambda_{\min}\bigl(\Sigma_{\beta_0,j\cup G^c}\bigr) \geq \Lambda_{\min}\bigl(\Sigma_0\bigr) \geq \frac{1}{C_e}.
\end{align}
By Condition~2(v), we can find $n_0 \in \mathbb{N}$ such that $a_n \leq 1/(2C_4C_e)$ for $n \geq n_0$.  Then from~\eqref{rate.tau0} and~\eqref{Eq:taulower}, we obtain on $\mathcal{T}$ that for $n \geq n_0$,
\[
  \min_{j\in G} \frac{\hat\tau_j^2}{\tau_j^2} \geq 1 - \max_{j\in G} \frac{|\hat\tau_j^2 - \tau_j^2|}{\tau_j^2} \geq 1 - C_4C_eK^2\sqrt{s \log p/n} \geq \frac{1}{2}.
\]
Using \eqref{gamma.decomp}, we conclude that on $\mathcal T$,
 \begin{eqnarray}
	\label{r1jp}
	\max_{j\in G}\|    \hat\Gamma_j /\hat\tau_j -   \Gamma_{0,j} /\tau_j \|_1 \leq C_e C_3 K^2s\sqrt{\log p/n} + C_e^2C_4K^2\sqrt{s\log p/n}.
	\end{eqnarray}
Consequently, combining  \eqref{r1j} and \eqref{r1jp}, there exists a constant such that it holds  on $\mathcal T$  
\begin{eqnarray}
\label{dTj1}
\max_{j\in G}|\text{rem}_{1,j}|  = \mathcal{O}\bigl(K^3 s\, {\log p}/\sqrt{n}\bigr) = \mathcal{O}(a_n).
\end{eqnarray}

\vskip 0.3cm
\noindent
\textbf{Remainder $\text{rem}_{2,j}$:}
By the mean value theorem, for each $i=1,\dots,n$, there exists $\alpha_i\in[0,1]$ such that 
$$\mu(x_{i,-G}^T\beta_{0,-G}) - \mu(x_{i,-G}^T\hat\beta_{-G}) = \mu'(x_{i,-G}^T\tilde\beta_{(i)})x_{i,-G}^T(\beta_{0,-G} - \hat\beta_{-G}),$$ 
where $\tilde\beta_{(i)} := \alpha_i \hat\beta_{-G} + (1-\alpha_i)\beta_{0,-G}.$ 
 Consequently, using H\"older's inequality and the KKT conditions from the optimization problem in Algorithm  \ref{grptest}, we obtain 
\begin{align} \nonumber
|\text{rem}_{2,j}| 
&= \sum_{i=1}^n \bigl|\hat w_{j,i} D^{-1}_{\hat\beta_{-G},ii}\mu'(X_{i,-G}^T\tilde \beta_{(i)})X_{i,-G}^T(\beta_{0,-G}-\hat\beta_{-G})\bigr| \\\nonumber
&\leq
\bigl|\hat w_j^T D^{-1}_{\hat\beta_{-G}} D^2_{\hat\beta_{-G}} X_{-G}(\beta_{0,-G}-\hat\beta_{-G})\bigr|+
\text{rem}_{3,j}
\\\nonumber
&\leq
\|\hat w_j^T D^{}_{\hat\beta_{-G}} X_{-G}\|_\infty \|\beta_{0,-G}-\hat\beta_{-G}\|_1 +
\text{rem}_{3,j}
\\\label{dTj2}
&= \mathcal{O}\bigl(K\sqrt{\log p} \, s\sqrt{\log p/n}\bigr) +
\text{rem}_{3,j},
\end{align}
where  
\begin{align*} 
\text{rem}_{3,j} &:=
\biggl|\sum_{i=1}^n \hat w_{j,i} D^{-1}_{\hat\beta_{-G},ii} \bigl(\mu'(x_{i,-G}^T\tilde \beta_{(i)}) -\mu'(x_{i,-G}^T\hat\beta_{G})\bigr)x_{i,-G}^T(\beta_{0,-G}-\hat\beta_{-G}) \biggr|
\\
&\phantom{:}\leq
L \sum_{i=1}^n| \hat w_{j,i}|\bigl| D^{-1}_{\hat\beta_{-G},ii}\bigr| \bigl\{x_{i,-G}^T(\beta_{0,-G}-\hat\beta_{-G}) \bigr\}^2.
\end{align*}
By Condition \ref{grptest.cond}(ii), we have $\|X_{-G}\gamma_{0,j} \|_\infty\leq K$, so we obtain 
\begin{align}
\nonumber
\max_{j\in G}\|w_{j}\|_\infty &=
\max_{j\in G}\| D_{\beta_0}(X_j -X_{-j}\gamma_{0,j})\|_\infty/(\sqrt{n}\tau_j) 
\\\nonumber
&\leq 
\max_{j\in G}
\frac{1}{\sqrt{n}\tau_j}\bigl\{\|D_{\beta_0}X_j \|_\infty + \|D_{\beta_0}X_{-G}\gamma_{0,j}\|_\infty\bigr\}
\\\label{bound_on_w}
&\leq \frac{2C_eC_0^{1/2} K }{\sqrt{n}}.
\end{align}

\noindent
Since  $\mu'$ is Lipschitz and $\|\exi\|_\infty \leq K$, we obtain that on $\mathcal T$, 
\begin{eqnarray}
\label{Diirates}
\bigl|D_{\hat{\beta}_{-G},ii} ^2  - D_{\beta_0,ii}^2 \bigr| 
\leq
 L|x_{i,-G}^T(\hat\beta_{-G}-\beta_{0,-G})|
\leq  
 L\|x_i\|_\infty \|\hat\beta_{-G}-\beta_{0,-G}\|_1 \leq  L K s\lambda.
\end{eqnarray}

\noindent
Thus 
\begin{eqnarray}
\label{Diihatrates}
\max_{i=1,\dots,n}\bigl| D_{\hat\beta_{-G},ii}\bigr| \asymp 1
\end{eqnarray}
by Condition~2(iii) and Condition~2(v). 
Therefore, on $\mathcal{T}$,
it follows that
\begin{align*}
\bigl\|D_{\hat\beta_{-G}}(X_j - X_{-G}\hat\gamma_{j}) \bigr\|_\infty
 &\leq 
\bigl\|D_{\hat\beta_{-G}}(X_j - X_{-G}\gamma_{0,j})\bigr\|_\infty + \bigl\|D_{\hat\beta_{-G}}X_{-G}(\hat\gamma_{j} - \gamma_{0,j})\bigr\|_\infty
\\
&\leq
\mathcal O(K) + \|D_{\hat\beta_{-G}}X_{-G}\|_\infty \|\hat\gamma_{j} - \gamma_{0,j}\|_1
\\
&= \mathcal O(K) + \mathcal O(K ^3 s\sqrt{\log p/n}) = \mathcal O(K).
\end{align*}
Consequently, on $\mathcal{T}$,
\begin{eqnarray}
\label{hatwMB}
\max_{j\in G}\|\hat w_{j}\|_\infty = 
\max_{j\in G}
\|D_{\hat\beta_{-G}}(X_{j} - X_{-G}\hat\gamma_{j})\|_\infty/(\sqrt{n}\hat\tau_j) = \mathcal O\left(\frac{K}{\sqrt{n}}\right).
\end{eqnarray}
Then using the result above, on $\mathcal T$, we obtain
\begin{align}
  \label{dTj3}
\nonumber 
\text{rem}_{3,j}
&\leq 
L\sum_{i=1}^n| \hat w_{j,i}|\bigl| D^{-1}_{\hat\beta_{-G},ii}\bigr| \bigl\{x_{i,-G}^T(\beta_{0,-G}-\hat\beta_{-G}) \bigr\}^2 \\
&= \mathcal{O}\bigl(K\sqrt{n}s\lambda^2\bigr) = \mathcal{O}(a_n).
\end{align}
The result follows from~\eqref{dTj}, together with the bounds~\eqref{eq:Tc}, \eqref{dTj1},~\eqref{dTj2} and \eqref{dTj3}.
\end{proof}

\begin{proof}[Proof of Theorem \ref{mbcor}]
We want to show that the quantiles of our test statistic for group testing, 
$$T:=\max_{j\in G} \left|\sum_{i=1}^n \hat w_{j,i} \hatRGi\right| $$ 
can be approximated by quantiles of its bootstrapped version 
$$W := T^{b} = \max_{j\in G} \left|\sum_{i=1}^n \hat w_{j,i} \hatRGi e_i^ b\right|,$$
where  $(e_i^b)_{i=1}^n$ is a sequence of independent and identically distributed $\mathcal N(0,1)$ random variables.
We can apply Theorem \ref{corMB} from Section \ref{sec:MB} together with Proposition \ref{grouptest.prop}. 
Adopting the notation of Theorem \ref{corMB} we let 
$$
T_0:=\max_{j\in G} \left|\sum_{i=1}^n  w_{j,i} \varepsilon_i\right|
,$$
and
$$
W_0 :=\max_{j\in G} \left|\sum_{i=1}^n w_{j,i} \varepsilon_i e_i^b\right|
.$$
Note that the maxima above can be rewritten without the absolute values using the fact that for any $a\in \mathbb R$ it holds that 
$|a| =\max\{a,-a\}$. 
Thus for $i=1,\ldots,n$ and $j\in G$ we let $x_{ij}: = \sqrt{n}w_{j,i}\varepsilon_i $
and $\hat x_{ij}:= \sqrt{n}\hat w_{j,i}\hatRGi$.
Moreover, for $i=1,\dots,n$ and $j\in G$ we also define $x_{i(j+|G|)}:= -x_{ij}$
and $\hat x_{i(j+|G|)}:= -\hat x_{ij}$.
We will apply Theorem \ref{corMB} with
 $x_{ij}$ and $\hat x_{ij}$ 
where $i=1,\dots,n$ and $j\in H:= G\cup \{j+|G|: j\in G\}.$


We now check that conditions \eqref{MB1}, \eqref{MB2}, \eqref{TandW}, \eqref{MB3} and \eqref{MB4} needed for Theorem \ref{corMB} are satisfied.
\\
{\bf Checking condition \eqref{MB1}}:\\
First, by the tower property, we have 
$$\s \mathbb E x_{ij}^2 = \sum_{i=1}^n\mathbb E \bigl\{w_{j,i}^2 \mathbb E(\varepsilon_i^2|X)\bigr\} .$$
{By Condition \ref{grptest.cond}(iii), it follows that $c_0\leq D_{\beta_0,ii}^2 \leq C_0$.} Consequently,  and using Condition \ref{grptest.cond}(i),
 there exist constants $c,C$ such that $c
\leq \mathbb E(\varepsilon_i^2|X)=\mathbb E(\eta_i^2/D_{\beta_0,ii}^2|X) \leq C.$
It follows that 
\begin{equation}
\label{eb1}
c \sum_{i=1}^n\mathbb E w_{j,i}^2
 \leq 
\s \mathbb E x_{ij}^2 
=
 \s n\mathbb E \bigl\{w_{j,i}^2 \mathbb E(\varepsilon_i^2|X)\bigr\} \leq 
C \sum_{i=1}^n\mathbb E w_{j,i}^2.
\end{equation}
Now recalling that 
$$w_j = D_{\beta_0} (X_j - X_{-G}\gamma_{0,j}) /(\sqrt{n} \tau_j),$$ 
we see that 
\begin{equation}
\label{eb2}
\sum_{i=1}^n \mathbb E w_{j,i}^2 = \mathbb E \|w_j\|_2^2 = 1.
\end{equation}
Therefore, combining \eqref{eb1} and \eqref{eb2}, we obtain
$c\leq \s \mathbb E x_{ij}^2 \leq C$ for $j \in H$, as required.
\\
{\bf Checking condition \eqref{MB2}:}\\
Recall from~\eqref{bound_on_w} that we have the deterministic bound
$$\max_{i,j}|w_{j,i}| =\mathcal O(K/{\sqrt{n}}).$$ 
Using this bound, we will now check that for suitable $B_n \asymp K$,
\begin{equation}
  \label{Eq:Cherno}
  \max_{k=1,2}\frac{1}{n}\sum_{i=1}^n \mathbb E |x_{ij}|^{2+k}/B_n^k
+\mathbb E\max_{j\in H}|x_{ij}/B_n|^{4}\leq
4.
\end{equation}
First observe that
\begin{align*}
\max_{k=1,2}\frac{1}{n}\sum_{i=1}^n \mathbb E |x_{ij}|^{2+k}/B_n^k 
&\lesssim 
 \max_{k=1,2} n^{k/2} \sum_{i=1}^n \mathbb E w_{j,i}^{2+k} /B_n^k 
\\
&\lesssim  \max_{k=1,2} n^{k/2} \left(\frac{K}{\sqrt{n}}\right)^k\sum_{i=1}^n \mathbb E w_{j,i}^{2} /B_n^k 
\\
&\lesssim  \max_{k=1,2} (K /B_n)^k ,
\end{align*}
and 
$$\mathbb E\max_{j\in H}|x_{ij}/B_n|^{4} \lesssim (K/B_n)^4.$$
Taking sufficiently large  $B_n\asymp K$, we can therefore guarantee that~\eqref{Eq:Cherno} holds, as required.

\noindent
{\bf Checking condition \eqref{TandW}:}\\
By Proposition \ref{grouptest.prop}, there exists a constant $C' > 0$ such that
$$\mathbb{P}(|T-T_0| > \zeta_1) \leq \mathbb{P}\left(\max_{j\in H} |T_j - w_j^T \varepsilon| > \zeta_1'\right) \leq \zeta_2',$$
for $\zeta_1' := C'K^3s{\log p}/\sqrt{n} $ and $\zeta_2' := 3/p$.  Next note that 
$$|W-W_0|\leq \max_{j\in H} \left| \frac{1}{\sqrt{n}}\sum_{i=1}^n (x_{ij} - \hat x_{ij})e_i^b\right|.$$
Now conditional on $(x_{ij})_{i=1}^n$ and $(\hat x_{ij})_{i=1}^n$ we have that $Z_j:= \frac{1}{\sqrt{n}}\sum_{i=1}^n (x_{ij} - \hat x_{ij})e_i^b\sim \mathcal N(0,\sigma_j^2)$ where $\sigma_j^2:= \frac{1}{n}\sum_{i=1}^n (x_{ij} - \hat x_{ij})^2.$ 
Therefore, 
$$\mathbb E_e |W-W_0| \leq \mathbb E_e \max_{j\in H} |Z_j|\leq \sqrt{2\log(|H|+1)}\max_{j\in H}\sigma_j 
.$$
Then it follows by Borell's inequality for any $t>0$,
$$\mathbb{P}_e\biggl(|W-W_0| > t +\mathbb E_e \max_{j\in H} |Z_j|\biggr) \leq \mathbb{P}_e\biggl(\max_{j\in H} |Z_j| >t +\mathbb E_e \max_{j\in H} |Z_j|\biggr) \leq e^{-t^2/(2\max_j\sigma_j^2)}.$$
Taking $t := \sqrt{2\log (2p+1)}\max_{j\in H}\sigma_j $ and noting that $|H|=2|G| \leq 2p,$
we obtain 
\begin{eqnarray}
\label{tec1}
\mathbb{P}_e\left(|W-W_0| > 2\sqrt{2\log (2p+1)}\max_{{j\in H}}\sigma_j  \right) \leq e^{-\log ({2}p+1)} \leq 1/p.
\end{eqnarray}
Denote $\Delta_2:= \max_{j\in H} \sigma_j^2.$
Then by Lemma \ref{MBdelta} in Appendix \ref{sec:appendix} there exists a constant $C'' > 0$
 such that 
\begin{eqnarray}
\label{tec2}
\mathbb{P}\bigl(\Delta_2 \geq  (C'')^2 K^6 s^2\lambda^2\bigr) \leq  4/p.
\end{eqnarray}
Therefore, combining \eqref{tec1} and \eqref{tec2}
$$\mathbb{P}\left(\mathbb{P}_e(|W-W_0| > \sqrt{2\log (2p)}C''K^3s\lambda
 ) > 1/p\right)\leq  4/p.$$
So we can take 
$$\zeta_1:= \max\{\sqrt{2\log (2p)}C''K^3s\lambda,\;\zeta_1'\} \asymp K^3s\log p/\sqrt{n},$$
 and $\zeta_2:= 4/p$ (in applying Theorem \ref{corMB}).
\\
{\bf Checking conditions \eqref{MB3} and \eqref{MB4}}:\\
Finally, by assumption \eqref{raten}, there exist constants $C_2',c_2 > 0$ such that
$$\zeta_1 \log (2|G|) + \zeta_2 \leq C_2' n^{-c_2},$$ and $$B_n^4 \log (2|G|n)^7 /n \leq C_2' n^{-c_2},$$
where $B_n\asymp K$.
\end{proof}

\subsubsection{Proofs for Section \ref{sec:power}}
\label{sec:power.proof}

\begin{proof}[Proof of Theorem \ref{glmnonlin}]
In this proof, it is convenient to write $\mathbb{P}_A(\cdot)$ and $\mathbb{E}_A(\cdot)$ as shorthand for $\mathbb{P}(\cdot|Z_A)$ and $\mathbb{E}(\cdot|Z_A)$ respectively.  Consider the decomposition
\[
\label{dec1}
T=
\hat w_A^T \hatRA =\hat w_A^T \hat D_{A}^{-1} \bigl(Y-\bmu(X\hat\beta)\bigr) 
=
\phi + \del + \er,
\]
where
$$\mathbf \phi := { \hat w_A^T \hat D_{A}^{-1}\bigl\{Y - \bmu\bigl(f_0(X)\bigr)\bigr\} }
,$$
$$\del := { \hat w_A^T \hat D_{A}^{-1} \bigl\{\bmu\bigl(f_0(X)\bigr) - \mu(X\beta_0)\bigr\}}
,$$
$$\er := {\hat w_A^T \hat D_{A}^{-1} \bigl\{\mu(X\beta_0) - \mu(X\hat\beta)\bigr\}}
.$$
There are three terms:
\begin{enumerate}[I.]
\item
The term $\phi$ is the pivot.  By the Berry--Esseen theorem, we will show below that (after scaling) it is well approximated by a normal random variable.
\item
The term $\del$ captures the deviation from the null hypothesis. If the null hypothesis is true, then $\del=0$. 
\item
The term $\er$ is a stochastic remainder term, for which we will develop a probabilistic bound below.
\end{enumerate} 
Let $\sigma := \|D_Y\hat{D}_A^{-1}\hat{w}_A\|_2$.  Then
\begin{align}
  \label{Eq:Decomp}
  \sup_{z\in \mathbb{R}} |\mathbb{P}_A(\phi + \er \leq z) - \Phi(z)| \leq \sup_{z\in \mathbb{R}} |\mathbb{P}_A(\phi + \er \leq z ) - \Phi(z/\sigma)| + \sup_{z\in \mathbb{R}}|\Phi(z/\sigma) - \Phi(z)|. 
\end{align}
Now, for any $\epsilon > 0$ and $z \in \mathbb{R}$,
\begin{align}
  \label{Eq:Decomp2}
  \mathbb{P}_A(\phi + \er \leq z) &\leq \mathbb{P}_A(\phi \leq z + \epsilon) + \mathbb{P}_A(|\er| >\epsilon) \nonumber \\
                                                        &\leq \Phi(z/\sigma) + \sup_{x \in \mathbb{R}} |\mathbb{P}_A(\phi \leq x) - \Phi(x/\sigma)| + \frac{\epsilon}{\sqrt{2\pi}\sigma} + \mathbb{P}_A(|\er| >\epsilon).
\end{align}
Similarly,
\begin{align}
  \label{Eq:Decomp3}
  \mathbb{P}_A(\phi + \er \leq z) &\geq \mathbb{P}_A(\phi \leq z - \epsilon) - \mathbb{P}_A(|\er| >\epsilon) \nonumber \\
                                                        &\geq \Phi(z/\sigma) - \sup_{x \in \mathbb{R}} |\mathbb{P}_A(\phi \leq x) - \Phi(x/\sigma)| - \frac{\epsilon}{\sqrt{2\pi}\sigma} - \mathbb{P}_A(|\er| >\epsilon).
\end{align}
Therefore, combining \eqref{Eq:Decomp},~\eqref{Eq:Decomp2} and~\eqref{Eq:Decomp3}, we find that
\begin{equation}
\label{eee4}
\sup_{z \in \mathbb{R}} \left|\mathbb{P}_A\left(T  - \del \leq z\right) - \Phi(z)\right| 
\leq \sup_{x \in \mathbb{R}} |\mathbb{P}_A(\phi \leq x) - \Phi(x/\sigma)| + \frac{\epsilon}{\sqrt{2\pi}\sigma} + \mathbb{P}_A(|\er| >\epsilon) + |\mathrm{rem}_0|,
\end{equation}
where $\mathrm{rem}_0 := \frac{1}{\sqrt{2\pi}}\bigl(\frac{1}{\sigma}-1\bigr)$.

\noindent
\textbf{Bound for the pivot.}\\
\noindent
We apply the Berry--Esseen theorem for non-identically distributed summands to $Z_0 := \phi/\sigma$. 
Note that
$$Z_0=
\frac{\sum_{i=1}^n \hat w_{A,i} \hat D_{A,ii}^{-1} \bigl\{Y_i - \mu\bigl(f_0(\exi)\bigr)\bigr\}}{ \sqrt{\sum_{i=1}^n \hat w_{A,i}^2 \hat D_{A,ii}^{-2}D_{Y,ii}^2 } } .
$$
For $i=1,\dots,n$, denote
$U_i := \hat w_{A,i} \hat D_{A,ii}^{-1} \bigl\{ Y_i - \mu\bigl(f_0(\exi)\bigr)\bigr\}$ 
and 
$\sigma_i^2:=\text{Var}(U_i|{Z_A}) = \hat w_{A,i}^2 \hat D_{A,ii}^{-2} D_{Y,ii}^2.$ 
Since $\mathbb E_A(U_i)=0$, the Berry--Esseen theorem \citep{esseen1942liapunov} yields that
$$\sup_{x \in \mathbb{R}} |\mathbb{P}_A(\phi \leq x) - \Phi(x/\sigma)| = \sup_{x\in\mathbb R}|\mathbb{P}_A( Z_0 \leq x) - \Phi(x)| \leq 
C_1 \frac{\sum_{i=1}^n \mathbb{E}_A\bigl(|U_i|^3\bigr)}{\bigl\{\sum_{i=1}^n \sigma_i^2\bigr\}^{3/2}}
,$$
where $C_1 > 0$ is a universal constant. 
Hence using Condition~1, 
\begin{equation}
\label{bebound}
\sup_{x\in\mathbb R}|\mathbb{P}_A( \phi \leq x) - \Phi(x/\sigma)| 
\leq 
\frac{C_1C_{\varepsilon}}{\sigma^3}  
\sum_{i=1}^n |\hat w_{A,i} D_{Y,ii}^{} \hat D_{A,ii}^{-1}|^3 \leq
\frac{C_1  C_{\varepsilon}}{\sigma}\|D_{Y}^{}  \hat D_{A}^{-1}  \hat w_{A}\|_\infty
.
\end{equation}
\textbf{Bound for} $\erzero$.
To bound $|\sigma-1|$,  
we first bound $\|D_Y^2 \hat D_A^{-2} - I \|_\infty$. First, since $\hat\beta_A \in \Theta(\lambda,\beta_0,X_A)$, $\mu'$ is Lipschitz and $\|\exi\|_\infty \leq K_X$, we obtain
\begin{align*}
|\hat D_{A,ii} ^2  -D_{\beta_0,ii}^2 | 
&=
|\mu'(x_i^T\hat\beta_A) - \mu'(x_i^T\beta_0)| 
\leq 
 L|x_i^T(\hat\beta_A-\beta_0)|
\\
&\leq  
 L\|x_i\|_\infty \|\hat\beta_A-\beta_0\|_1 \leq 
 L K_X s\lambda_A.
\end{align*}
Therefore, $\hat D_{A,ii}^2 \geq D_{\beta_0,ii}^2/2$ under the condition $6\domin^{-2}L K_X s\lambda_A \leq 1/2.$
This then implies that 
$$|\eta_{i,2}|:= |D_{\beta_0,ii}^2/\hat D_{A,ii} ^2  -1 | \leq 2\domin^{-2} L K_X s\lambda_A.$$
Next, by assumption, we have 
$$|\eta_{i,1}| := |D_{Y,ii}^2  D_{\beta_0,ii}^{-2} - 1| \leq 2\domin^{-2} L K_X s\lambda_A,$$ 
(note that under $H_0,$ it holds that $D_Y = D_{\beta_0}$, so $|\eta_{i,1}| = |D_{Y,ii}^2  D_{\beta_0,ii}^{-2} - 1|=0$ and the required bound trivially holds).  Then 
\begin{align*}
\|D_Y^2 \hat D_A^{-2} - I \|_\infty &=
\max_{i=1,\dots,n} \bigl|D_{Y,ii}^2 \hat D_{A,ii}^{-2} - 1\bigr| 
\\
&=
\max_{i=1,\dots,n}
\left|
\left(D_{Y,ii}^2/ D_{\beta_0,ii}^2 -1 + 1\right)
\left(D_{\beta_0,ii}^2/\hat D_{A,ii}^2 -1 + 1\right) - 1
\right|,
\\
&= 
\max_{i=1,\dots,n}
\left|
\left(\eta_{i,1} + 1\right)
\left(\eta_{i,2} + 1\right) - 1
\right|,
\\
&= 
\max_{i=1,\dots,n}
\left|
\eta_{i,1}\eta_{i,2} + \eta_{i,1} + \eta_{i,2} 
\right|,
\\
&\leq
 4\domin^{-2} L K_X s\lambda_A + (2\domin^{-2} L K_X s\lambda_A )^2.
\end{align*}
Finally, by assumption $6\domin^{-2} L K_X s\lambda_A\leq 1/2$ and from the last display 
and Lemma \ref{Dconvergence}, it follows that 
$$|\sigma-1|\leq 
\|D_Y^2 \hat D_A^{-2} - I \|_\infty \leq 6\domin^{-2} L K_X s\lambda_A =: r_{\erzero}.$$
We also see from this calculation that under our conditions, $\sigma \geq 1/2$ and $\|D_Y \hat D_A^{-1}\|_\infty \leq 2$.


\vskip 0.5cm
\noindent
\textbf{Bound for $\er$.}  
A Taylor expansion of $\mu$ yields
$$\mu(x_i^T\beta_0) - \mu(x_i^T\hat\beta) = \mu'(x_i^T\tilde \beta_{(i)})x_i^T(\beta_0-\hat\beta),$$
where $\tilde \beta_{(i)} = \alpha_{i} \beta_0 + (1-\alpha_{i})\hat\beta$ for some $\alpha_{i} \in [0,1]$.  
Let $ D_{\tilde\beta}$ denote a diagonal matrix with diagonal entries 
$ D_{\tilde\beta,{ii}}^2:=\mu'(x_i^T\tilde \beta_{(i)}).$ 
Then
$$\er = \hat w_A^T \hat D_A^{-1}\bigl\{\mu(X\beta_0) - \mu(X\hat\beta)\bigr\}=\hat w_A^T \hat D_A^{-1} D_{\tilde\beta}^2X(\beta_0-\hat\beta)=:
\erone + 
\ertwo
,$$
where 
$$\erone = \hat w_A^T  \hat D_{A} X_{}(\beta_0-\hat\beta),$$
and  $$\ertwo = {\hat w_A^T \hat D_A^{-1} ( D_{\tilde\beta}^2 -  \hat D_{A})X_{}(\beta_0-\hat\beta)}.$$
Using  H\"older's inequality together with 
$\hat\beta\in \Theta(\lambda,\beta_0,X)$ and the KKT conditions of the square-root Lasso \eqref{sqrt.lasso}, 
we have 
$$|\erone|= \|\hat w_A^T \hat D_{A} X_{}(\beta_0-\hat\beta) \|_\infty 
\leq 
\|\hat w_A^T \hat D_A X_{}\|_\infty \|\beta_0-\hat\beta\|_1 
\leq 
\lambda_{\text{sq}}\sqrt{n} s\lambda .$$
To bound the second term, $\ertwo$, first 
by the Lipschitz property of $\mu'$ we have
\begin{equation}
\label{lipbound}
|\mu'(x_i^T\tilde\beta_{(i)})  - \mu'(x_i^T\hat\beta_A)| 
\leq L |x_i^T(\tilde \beta_{(i)} - \hat\beta_A)|
\leq L \bigl( |x_i^T(\hat\beta -\beta_0)| + |x_i^T(\beta_0 - \hat\beta_A)|\bigr)
,
\end{equation}
Then, on the event that $\hat\beta\in \Theta(\lambda,\beta_0,X)$,
\begin{align*}
\nonumber
|\ertwo|
&:=
\bigl|\hat w_A^T \hat D_A^{-1} ( D_{\tilde\beta}^2- \hat D_{A}^2) X (\beta_0 - \hat{\beta})\bigr|
\\\nonumber
& = 
\biggl| \sum_{i=1}^n  \hat w_{A,i} \hat D_{A,ii}^{-1}\Bigl\{\mu'(x_i^T\tilde \beta_{(i)}) - \mu'(x_i^T\hat\beta_A)\Bigr\} x_i^T (\beta_0 - \hat{\beta})\biggr|
\\\nonumber
& \leq  
 \sum_{i=1}^n {{\bigl|\hat w_{A,i} \hat D_{A,ii}^{-1}\bigr|} }
\bigl|\mu'(x_i^T\tilde \beta_{(i)}) - \mu' (x_i^T\hat\beta_A)\bigr| \bigl| x_i^T (\beta_0 - \hat{\beta})\bigr| 
\\\nonumber
& \stackrel{\eqref{lipbound}}{\leq} 
L\sum_{i=1}^n {|\hat w_{A,i} \hat D_{A,ii}^{-1}|}  \left( \frac{3}{2}|x_i^T(\beta_0 - \hat{\beta})|^2 + \frac{1}{2}|x_i^T(\beta_0 - \hat\beta_A)|^2\right)
\\\nonumber
& \leq  
L \max_{i=1,\dots,n}{\bigl|\hat w_{A,i} \hat D_{A,ii}^{-1}\bigr|}\sum_{i=1}^n \left( \frac{3}{2}|x_i^T(\beta_0 - \hat{\beta})|^2 + \frac{1}{2}|x_i^T(\beta_0 - \hat\beta_A)|^2\right)
\\
&\leq 
4L \domin^{-1} \|\hat w_A\|_\infty {s}(\lambda^2 + \lambda_{A}^2) {n}.
\end{align*}
Therefore,
\label{bo1}
$\mathbb{P}_A(|\er | \geq r_{\er})\leq \delta$,
 where 
$$ r_{\er} :=  \lambda_{\text{sq}} \sqrt{n} s\lambda + 4L \domin^{-1} \|\hat w_A\|_\infty {s}(\lambda^2 + \lambda_{A}^2) {n}.
$$

\noindent
We conclude, using \eqref{eee4} and \eqref{bebound}, and taking $\epsilon:= r_{\er}$ that
\begin{align*}
\label{eee5}
\left|\mathbb{P}_A\left(T  - \del < z\right) - \Phi(z)\right| 
&\leq 
 \frac{C_1  C_{\varepsilon}}{\sigma}\|D_{Y}\hat D_{A}^{-1}  \hat w_{A}\|_\infty + \frac{\sqrt{2}r_{\mathrm{rem}_1}}{\sqrt{\pi}} + \delta + \frac{\sqrt{2}r_{\mathrm{rem}_0}}{\sqrt{\pi}}
\\
&\leq 4C_1  C_{\varepsilon}\|\hat{w}_{A}\|_\infty + \frac{\sqrt{2}r_{\mathrm{rem}_1}}{\sqrt{\pi}} + \delta + \frac{\sqrt{2}r_{\mathrm{rem}_0}}{\sqrt{\pi}}.
\end{align*}
The result follows.
\end{proof}

\subsubsection{Proofs for Section \ref{sec:cons.logistic}}
\label{Sec:LogisticProofs}

The logistic loss function is
$$\rho(u,y) := -yu + d(u),$$
where $d(\xi):= \log \bigl(1+ e^\xi\bigr)$, and we let $f_\beta(x) := x^T\beta$.  We define the risk function 
$$R(f|X) := \s \mathbb E\bigl\{\rho\bigl(f(x_i),Y_i\bigr) \bigm|X\bigr\}$$ 
and set 
$\beta_0 := \argmin_{\beta \in \R^p}R(f_\beta|X)$.  

\begin{proof}[Proof of Corollary~\ref{cor.logistic}]
We apply Theorem \ref{glmnonlin} to the case of logistic regression to obtain local guarantees on the power of the test.  To this end, we need to bound $\delta$ in~\eqref{consistency.lasso2} and Condition~\ref{model2} of Theorem~\ref{glmnonlin}. 

To bound $\delta$, we note that by Lemma~\ref{logit.rates} in Section~\ref{sec:appendix} 
 with $t:=\log (2p)$, we have  
with probability at least $1-1/(2p)$ that
$$R\bigl( f_{\hat\beta}|X\bigr) - R\bigl(f_0|X\bigr) + \lambda \|\hat\beta-\beta_{0}\|_1 \leq  \frac{17\lambda^2{s}(e^{\eta}/\epsilon_0+1)^2}{ \phi^2}.$$
In what follows we work on the event where this occurs.
We next want to obtain a bound on $ \|X(\hat\beta-\beta_{0})\|_2^2/n$. 
Note that the second derivative of the loss function is
$$\frac{\partial^2 \rho(u,y)}{\partial u^2}= d''(u)=\frac{e^u}{1+e^u}\left( 1 - \frac{e^u}{1+e^u} \right).$$ 
For $\|x\|_\infty \leq K$ and any $f$ with $\sup_{x : \|x\|_\infty \leq K} |f(x) - f_0(x)| \leq \eta $, we therefore have
\begin{equation}
  \label{ddotrho}
d''(f(x)) \geq (e^{|f_0(x)|+\eta}+1)^{-2}\geq (e^{\eta} / \epsilon_0 + 1)^{-2}=:c_0^2 >0.
\end{equation}
Note that for any $\tilde{\beta}$ on the line segment between $\beta_0$ and $\hat{\beta}$, we have
\begin{align*}
\sup_{x : \|x\|_\infty \leq K} |f_{\tilde{\beta}}(x) - f_0(x)| &\leq \sup_{x : \|x\|_\infty \leq K} |f_{\tilde{\beta}}(x) - f_{\beta_0}(x)| +\eta/2 \\
&\leq K\|\hat{\beta} - \beta_0\|_1 + \eta/2 \leq \eta.
\end{align*}
Thus we can conclude using a Taylor expansion of the loss function that there exist $\{\tilde{\beta}_{(i)}: i=1,\ldots,n\}$, each on the line segment from $\beta_0$ to $\hat{\beta}$, such that
\begin{align*}
R( f_{\hat{\beta}}|X) - R(f_{\beta_0}|X) 
&=
 \frac{1}{2n} \sum_{i=1}^n d''(x_i^T\tilde{\beta}_{(i)}) \bigl(x_i^T(\hat\beta-\beta_0)\bigr)^2 
\\
&\geq
c_0^2 \|X(\hat\beta-\beta_0)\|_2^2/(2n)
.
\end{align*}
We deduce that there exists a constant $\tilde C > 0$ such that with $\lambda = \tilde{C} \sqrt{\log(2p)/n}$, we have $\delta = \mathbb{P}\bigl(\hat\beta \notin \Theta(\lambda,\beta_0,X)|X\bigr)\leq 1/(2p)$.
\par
It remains to check that Condition \ref{model2} of Theorem \ref{glmnonlin} is satisfied. Firstly, the inverse link function $\mu(u) = 1/(1+e^{-u})$ is differentiable and Lipschitz with constant $1$. 
Moreover, by
\eqref{ddotrho}, 
$D_{Y,ii}^2 \geq d_{\min}^2 := c_0^2$
and also $D_{\beta_0,ii}^2 \geq c_0^2$. 
Finally, observe that $\mathbb E\bigl\{|Y_i - \mu(f_0(x_i))|^3|X\bigr\} \leq 1$. Moreover, 
$12\domin^{-2}LK_Xs\lambda  = 12c_0^{-2}LK_Xs\lambda\leq 1$ by hypothesis. Therefore, Condition \ref{model2}  is satisfied.
\end{proof}

	\appendix
	\section{Appendix}
	\label{sec:appendix}

\subsection{Auxiliary lemmas}
\begin{lem}[Hoeffding's inequality for a maximum of $p$ averages]
\label{hoef}
 Suppose that for each $j =1,\ldots,p$, the random variables $Z_{1j},\ldots,Z_{nj}$ are independent with 
$$\mathbb E Z_{ij} = 0, \quad |Z_{ij}| \leq c_i.$$
Then for all $t>0$
$$\mathbb{P}\biggl(\max_{j=1,\dots,p} \biggl|\s Z_{ij}\biggr|^2 \geq \frac{\|c\|_2^2}{n} {\frac{{2(t+\log(2p))}}{n}}\biggr) \leq e^{-t}.$$
\end{lem}
\begin{proof}[Proof of Lemma \ref{hoef}]
Apply Corollary 17.1 in \cite{sf} together with a union bound.
\end{proof}

\begin{lem}
\label{Dconvergence}
Let $\tilde A,\tilde B\in \mathbb R^{n\times n}$ be diagonal matrices and suppose that $\tilde{B}$ is invertible.  Let $w\in\mathbb R^n$ satisfy $\|w \|_2 =1$.  Then 
$$\left| \frac{\|\tilde A w\|_2}{\|\tilde B w\|_2} - 1\right| \leq \|\tilde A^2\tilde B^{-2}-I\|_\infty.$$
\end{lem}

\begin{proof}[Proof of Lemma \ref{Dconvergence}]
$$|\| \tilde A w \|_2^2  - \| \tilde B w\|_2^2| =
|w^T (\tilde A^2- \tilde B^2)w| 
\leq \max_{i=1,\dots,n}\frac{|\tilde A_{ii}^2 - \tilde B_{ii}^2|}{\tilde B_{ii}^2} \|\tilde B w\|_2^2 = \|\tilde A^2\tilde B^{-2} - I\|_\infty \|\tilde B w\|_2^2.
 $$
Hence 
$$\left|\frac{\| \tilde A w\|_2 }{\|\tilde B w\|_2}-1\right| = \frac{\left|\| \tilde A w\|_2^2/\|\tilde B w\|_2^2-1\right|}{\left|\| \tilde A w\|_2/\|\tilde B w\|_2+1\right|}
\leq 
{\|\tilde A^2\tilde B^{-2} - I\|_\infty },$$
as required.
\end{proof}

\subsection{Auxiliary lemmas for Group Testing}
	
\begin{lem}
\label{prob.bounds}
Under the conditions of Proposition \ref{grouptest.prop}, we have 
$$\mathbb{P}(\mathcal T^c) \leq 3/p.$$

\end{lem}

\begin{proof}

\noindent
To obtain a probability bound for $\mathcal T_{1,j}$,
we can apply Lemma \ref{hoef}, noting that 
$$\frac{1}{n}\| u_j^T D_{\beta_0}X_{-G}\|_\infty = \max_{k\in G^c} \biggl|\s Z_{ijk}\biggr|,$$
 with $Z_{ijk}:=u_{j,i}D_{\beta_0,ii}X_{-G,i,k}$, where $X_{-G,i,k}$ is the $(i,k)$-th entry of the matrix $X_{-G}$
and $u_{j,i}$ is the $i$-th entry of $u_j.$ Note that by Condition \ref{grptest.cond} (ii), it follows that 
$|u_{j,i}|\leq 2K$  and by Condition \ref{grptest.cond} (iii), we have $|D_{\beta_0,ii}|\leq C_0$. Therefore, 
 $|Z_{ijk}|\leq c_i$ for $c_i:=2C_0K^2$ and for all $i,j,k$. 
Thus Lemma \ref{hoef} implies that for all $t>0,$
$$ \mathbb{P}\biggl(\frac{1}{n^2}\| u_j^T D_{\beta_0}X_{-G}\|_\infty^2 \geq 2(2C_0K^2)^2\frac{t+\log (2p)}{n}\biggr) \leq e^{-t}.$$
Therefore, 
\begin{equation}
\label{T1j}
\mathbb{P}(\mathcal T_{1,j}^c ) \leq 1/(2p)^{2}.
\end{equation}
\noindent
For the set $\mathcal T_{3,j}$, by the sub-Gaussianity of $\eta_i = Y_i - \mu(x_i^T\beta_0)$ from Condition \ref{grptest.cond} (i), there exists a constant $C>0$ such that
\begin{equation}
\label{T3j}
\mathbb{P}\left(\mathcal T_{3,j}^c\right) = \mathbb{P}\left(\max_{k\in G^c}\biggl|\s X_{j\cup G^c,i,k}(Y_i-\mu(x_i^T\beta_0))\biggr|   \geq CK\sqrt{\frac{\log (2p)}{n}}\right) \leq 1/p^{2}.
\end{equation}
Therefore, 
$$\mathbb{P}\left(\mathcal T_j^c\right)\leq \mathbb{P}(\mathcal T_{1,j}^c ) + \mathbb{P}(\mathcal T_2^c) + \mathbb{P}(\mathcal T_{3,j}) 
\leq 1/p^2 + \mathbb{P}(\mathcal T_2^c) + 1/p^2.$$
Using bounds \eqref{T1j} and \eqref{T3j}, the fact that $|G|\leq p$ and the assumption $\mathbb{P}(\mathcal T_2^c)\leq 1/p$, we obtain by a union bound that
$$\mathbb{P}\left(\cup_{j \in G} \mathcal T_j^c\right)\leq |G| \max_{j\in G} \mathbb{P}(\mathcal T_{1,j}^c) + \mathbb{P}(\mathcal T_{2}^c)+ |G|\max_{j\in G} \mathbb{P}(\mathcal T_{3,j}^c
)\leq 2/p + \mathbb{P}(\mathcal T_2^c) \leq 3/p.$$

\end{proof}

	\begin{lem}
	\label{rates.nodewise}
	Assume Conditions \ref{model} and \ref{grptest.cond}. 
For $j\in G$, we let $u_j := X_j - X_{-j}\gamma_{0,j}$ and define the sets 
$$\mathcal T_{1,j} := \left\{\|u_j^TD_{\beta_0} X_{-G}\|_\infty/n \leq 6 C_0 K^2\sqrt{\log (2p)/n}\right\}, $$
$$\mathcal T_{2}   := \left\{ \hat\beta_{-G}\in\Theta_{-G}(\lambda,\beta_0)\right\}.$$ 
Then there exist $\lambda_{\mathrm{nw}} \asymp \sqrt{\log p/n}$, and positive constants $C_3$ and $C_4$ such that on the set 
$\cap_{j\in G} ( \mathcal T_{1,j} \cap \mathcal T_{2}),$ it holds that
\begin{align*}
\max_{j\in G}\|\hat\gamma_j - \gamma_{0,j}\|_1 &\leq C_3 K^2s\sqrt{\log p /n},\\
\max_{j\in G}|\hat\tau_j^2 - \tau_j^2| &\leq C_4 K^2 \sqrt{s \log p/n},
\end{align*}
whenever $\hat\tau_j^2 := \| \hat D(X_j - X_{-G}\hat\gamma_j)\|_2^2/n  > 0$.

\end{lem}

\begin{proof}[Proof of Lemma \ref{rates.nodewise}]
To obtain rates of convergence for $\hat\gamma_j$ from Algorithm \ref{grptest}, we follow the arguments in the proof of Theorem~3.2 in \cite{vdgeer13}, which considers nodewise regression
with random bounded design. The difference is that they define a nodewise regression program with design matrix $X_{-j}$ and use the Lasso, whereas we want to apply the nodewise regression with a smaller design matrix, $X_{-G}$, and we use the square-root Lasso.  We also seek finite-sample, as opposed to asymptotic, bounds, but this requires only minor modifications.  But since we assume that $\hat \tau_j>0$, the square-root Lasso program with penalty $\lambda_{\text{nw}}$ corresponds to the Lasso program with penalty $\lambda_{\text{Lasso}}:= \hat\tau_j \lambda_{\text{nw}}.$
We now check that the appropriate finite-sample analogues of conditions (D1)--(D5) of Theorem 3.2 from \cite{vdgeer13} are satisfied for $\tilde X:=X_{j\cup {G^c}}$. 
Firstly, the analogues of (D1), (D2), (D4) are satisfied directly by the assumptions in Conditions \ref{model} and \ref{grptest.cond}.
For (D3), first note that the smallest eigenvalue of $\Sigma_{\beta_0,j\cup G^c} := \mathbb E \tilde X^T D_{\beta_0}^2 \tilde X /n $
 is lower bounded by the smallest eigenvalue of $\Sigma_{0} = \mathbb E X_{}^T D_{\beta_0}^2 X_{} /n $, which is in turn lower bounded by $1/C_{e}$. Similarly, 
$\|\Sigma_{\beta_0,j\cup G^c}\|_\infty \leq \|\Sigma_{\beta_0}\|_\infty \leq C_e,$.  Finally, Condition (D5) is satisfied on $\mathcal T_2$.

As in the proof of Proposition~\ref{grouptest.prop}, let $k_j$ denote the index of column $X_j$ in the matrix $X_{j\cup G^c}$.
We write 
$$\hat\Gamma_j := (-\hat\gamma_{j,1},\dots,-\hat\gamma_{j,k_j-1},1,-\hat\gamma_{j,k_j+1},\dots,-\hat\gamma_{j,|G^c|})^T,$$
and $\Gamma_{0,j}$ for its analogy defined in terms of $\gamma_{0,j}.$ Then note that we can write $X_{j}- X_{-G}\hat\gamma_j = \tilde X \hat\Gamma_j$
and recall that
$\hat\tau_j^2 = \|\hat D \tilde X \hat\Gamma_{j}\|_2^2/{n}$ and 
$\tau_j^2 = \mathbb E\|D_{\beta_{0}} \tilde X \Gamma_{0,j}\|_2^2/{n}$.
By inspecting the proof of Theorem~3.2 of \cite{vdgeer13}, we conclude that 
there exist positive constants $C_3$ and $C_4$ such that on 
$\cap_{j\in G} (\mathcal T_{1,j} \cap \mathcal T_{2}),$
it holds that
\begin{align*}
\max_{j\in G}\|\hat\gamma_j - \gamma_{0,j}\|_1 &\leq  C_3 K^2s\sqrt{\log p /n},\\
\max_{j\in G}|\hat\tau_j^2 - \tau_j^2| &\leq  C_4 K^2 \sqrt{s \log p/n},
\end{align*}
as required.
\end{proof}
The following lemma bounds a term $\Delta_2$ defined in the proof of Theorem~\ref{mbcor}.
\begin{lem}
\label{MBdelta}
Under the conditions of Theorem \ref{mbcor},
there exists a constant $C'' > 0$
 such that 
\[
\mathbb{P}\bigl(\Delta_2 \geq  (C'')^2 K^6 s^2\lambda^2\bigr) \leq  4/p.
\]

\end{lem}	
	
\begin{proof}[Proof of Lemma \ref{MBdelta}]

On the set $\mathcal T$ defined in the proof of Proposition \ref{grouptest.prop}, we have
\begin{align}
\nonumber
\Delta_2 &=
 \max_{{j\in H} }\sigma_j^2 = \max_{j\in H} \frac{1}{n}\sum_{i=1}^n (x_{ij} - \hat x_{ij})^2\\
& =  
\max_{j\in H} \sum_{i=1}^n (w_{j,i}\varepsilon_i  - \hat w_{j,i}\hatRGi)^2
\\\nonumber
&\leq  
2\max_{j\in H} \sum_{i=1}^n (w_{j,i}-\hat w_{j,i})^2 \varepsilon_i^2
 + 
2\max_{j\in H} \sum_{i=1}^n \hat w_{j,i}^2( \hatRGi-\varepsilon_i)^2
\\
&=:
r_1+r_2.
\end{align}
Now we bound $r_1.$
By similar arguments as in the proof of Proposition \ref{grouptest.prop}, we will now show that  on $\mathcal T,$
$$\max_{i,j}(w_{j,i}-\hat w_{j,i})^2 = \mathcal O(K^3s^2\lambda^2 /n).$$
First,
\begin{align*}
\max_{i,j}|w_{j,i}-\hat w_{j,i}|
&= 
\max_{i,j}\frac{1}{\sqrt{n}} 
| D_{\beta_0,ii}e_i^T X_{j\cup G^c} \Gamma_{0,j}/\tau_j 
-
D_{\hat\beta_{-G},ii}e_i^T X_{j\cup G^c}  \hat\Gamma_{j}/\hat\tau_j  |
\\
&\leq 
\max_{i,j}
\frac{1}{\sqrt{n}}
 | (D_{\beta_0,ii} - D_{\hat\beta_{-G},ii})e_i^T X_{j\cup G^c}  \Gamma_{0,j}/\tau_j | 
\\&\hspace{1cm}+ 
\frac{1}{\sqrt{n}}
\max_{i,j} | D_{\hat\beta_{-G},ii}e_i^T X_{j\cup G^c}  (\hat\Gamma_{j}/\hat\tau_j -\Gamma_{0,j}/\tau_j) |.
\end{align*}
Now we can use Condition \ref{grptest.cond}(ii),~\eqref{Eq:Dbeta0}, \eqref{Diirates}, \eqref{Diihatrates} and \eqref{r1jp} to bound the terms in the last display and obtain
\begin{align*}
\max_{i,j}|w_{j,i}-\hat w_{j,i}|
&\leq 
\frac{1}{\sqrt{n}}\max_{i}\bigl| D_{\beta_0,ii} - D_{\hat\beta_{-G},ii}\bigr|
\max_{i,j}
| e_i^TX_{j\cup G^c}  \Gamma_{0,j}/\tau_j | 
\\
&\hspace{1cm}+ 
\frac{1}{\sqrt{n}}\max_{i,j}
\bigl| D_{\hat\beta_{-G},ii}\bigr| \|X_{j\cup G^c} \|_\infty \|\hat\Gamma_{j}/\hat\tau_j -\Gamma_{0,j}/\tau_j \|_1
\\
&\lesssim 
\frac{1}{\sqrt{n}}LK^2s\lambda + \frac{1}{\sqrt{n}} K^3s\sqrt{\log p/n}
\\
&\lesssim 
 \frac{1}{\sqrt{n}} K^3s\lambda
.
\end{align*}
Moreover, by the sub-Gaussianity of $\eta_i$ and since $D^2_{\beta_0,ii}\geq c_0$ (by Condition \eqref{grptest.cond}(iii)), 
there exist constants $C,C''' > 0$ such that with probability at least $1-1/p$, we have
\begin{eqnarray}
\label{etasq}
\frac{1}{n}\sum_{i=1}^n \varepsilon_i^2 \leq \s\mathbb E\varepsilon_i^2 + C \sqrt{\frac{\log p}{n}}\leq  C'''.
\end{eqnarray}
Therefore, with probability at least $1-1/p-\mathbb{P}(\mathcal T^c)$,
\begin{eqnarray}
\label{MBe2}
r_1 = \max_j  \sum_{i=1}^n (w_{j,i}-\hat w_{j,i})^2 \varepsilon_i^2
\leq
\max_{i,j}  (w_{j,i}-\hat w_{j,i})^2 \sum_{i=1}^n \varepsilon_i^2 
\lesssim
K^6s^2\lambda^2.
\end{eqnarray}

We now bound $r_2$.  Using Condition \ref{grptest.cond}(iii), together with the fact that $\hat\beta_{-G} \in \Theta_{-G}(\lambda,\beta_0)$ on $\mathcal{T}$ we have that on this event, 
\begin{align*}
|\hatRGi-\varepsilon_i |
	&=
  \bigl|D_{\hat\beta_{-G},ii}^{-1}\bigl(Y_i - \mu(X_{i,-G}\hat\beta_{-G})\bigr) - D_{\beta_0,ii}^{-1}\bigl(Y_i - \mu(X_{i,-G}\beta_{0,-G})\bigr) \bigr|
	\\
	&\leq 
	C_0\bigl|D_{\hat\beta_{-G},ii}^{-1}X_{i,-G}(\beta_{0,-G} - \hat\beta_{-G})\bigr| +
		\bigl|D_{\hat\beta_{-G},ii}^{-1} - D_{\beta_0,ii}^{-1}\bigr|\bigl|Y_i - \mu(X_{i,-G}\beta_{0,-G})\bigr|
	\\
	&\leq 
		C_0\bigl|D_{\hat\beta_{-G},ii}^{-1}X_{i,-G}(\beta_{0,-G} - \hat\beta_{-G})\bigr| 
		+
			D_{\beta_{0},ii}^{-1}
		\left|\frac{D_{\beta_0,ii}}{D_{\hat\beta_{-G},ii}}- 1\right||\eta_i|.
\end{align*}
Then using the fact that $\hat\beta_{-G} \in \Theta_{-G}(\lambda,\beta_0)$ on $\mathcal{T}$ and using that 
$$\biggl|\frac{D_{\beta_0,ii}}{D_{\hat\beta_{-G},ii}}- 1\biggr| \leq \biggl|\frac{D_{\beta_0,ii}^2}{D_{\hat\beta_{-G},ii}^2}- 1\biggr|
\leq 2c_0^{-1}LKs\lambda$$
 (which follows similarly as in the proof of Theorem \ref{glmnonlin}) and using \eqref{etasq}, we obtain that on $\mathcal{T}$,
\begin{align*}
\sum_{i=1}^n (\hatRGi-\varepsilon_i )^2
	&\leq 
2C_0 \sum_{i=1}^n
		D_{\hat\beta_{-G},ii}^{-2}|X_{i,-G}(\beta_{0,-G} - \hat\beta_{-G})|^2 +
			2\sum_{i=1}^n \left|\frac{D_{\beta_0,ii}}{D_{\hat\beta_{-G},ii}}- 1\right|^2 \varepsilon_i^2
			\\
			&\lesssim 
		s\lambda^2n +
			K^2s^2\lambda^2n \lesssim K^2 s^2\lambda^2 n
			.
\end{align*}
Then
\begin{eqnarray}
\label{rate2}
r_2 = 2\max_{j\in H} \sum_{i=1}^n \hat w_{j,i}^2( \hatRGi-\varepsilon_i)^2
	\lesssim 
\frac{K^2}{n} K^2s^2\lambda^2n = K^4 s^2 \lambda^2
			,
\end{eqnarray}
where we used  $\|\hat w_{j}\|_\infty^2 = \mathcal O(K^2/n)$ (which follows from \eqref{hatwMB}).
\\

Overall, collecting \eqref{rate2} and \eqref{MBe2} there exists a constant $C''$
 such that 
\begin{eqnarray}
\label{tec22}
\mathbb{P}(\Delta_2 \geq  (C'')^2 K^6 s^2\lambda^2) \leq 1/p + \mathbb{P}(\mathcal T^c) \leq 4/p.
\end{eqnarray}

\end{proof}
	
\subsection{Multiplier bootstrap}
\label{sec:MB}
We summarize Corollary 3.1 from \cite{chernozhukov2013gaussian}. To this end, we need the following condition.
\begin{cond}
\label{TWapprox}
Let $(x_{i})_{i=1}^n$ be $n$ independent random vectors with values in $\mathbb R^g$ 
satisfying 
\begin{equation}
\label{MB1}
c_1\leq \frac{1}{n}\sum_{i=1}^n\mathbb E x_{ij}^2 \leq C_1
\end{equation} 
and
\begin{equation}
\label{MB2}
\max_{k=1,2}\frac{1}{n}\sum_{i=1}^n \mathbb E |x_{ij}|^{2+k}/B_n^k
+\mathbb E\max_{1\leq j\leq g}|x_{ij}/B_n|^{4}\leq 4.
\end{equation}
 Define
$$T_0 := \max_{1\leq j \leq g}\frac{1}{\sqrt{n}}\sum_{i=1}^n x_{ij}.$$
Let $(e_i)_{i=1}^n$ be a sequence of i.i.d. $\mathcal N(0,1)$ random variables independent of $(x_{i})_{i=1}^n$ and define
$$W_0:= \max_{1\leq j \leq g}\frac{1}{\sqrt{n}}\sum_{i=1}^n x_{ij}e_i.$$
Assume that there exist $\zeta_1,\zeta_2\geq 0$ such that
\begin{equation}
\label{TandW}
\mathbb{P}(|T - T_0| > \zeta_1 ) \leq \zeta_2,
\quad\quad \mathbb{P}(\mathbb{P}_e(|W - W_0| > \zeta_1) >\zeta_2 ) < \zeta_2,
\end{equation}
where $\mathbb{P}_e$ is the probability measure induced by the multiplier variables $(e_i)_{i=1}^n$ holding $(x_{i})_{i=1}^n$ fixed.

\end{cond}

\begin{thm}[Corollary 3.1 in \cite{chernozhukov2013gaussian}]
\label{corMB}
Suppose that Condition \ref{TWapprox} is satisfied with 
\begin{equation}
\label{MB3}
\zeta_1\sqrt{\log g} + \zeta_2 \leq C_2 n^{-c_2}
\end{equation}
and 
\begin{equation}
\label{MB4}
B_n^4 \log(gn)^7 /n\leq C_2 n^{-c_2}.
\end{equation}
Then there exist constants $c,C>0$ depending only on $c_1,C_1,c_2$ and $C_2$ such that
$$\sup_{\alpha\in(0,1)} |\mathbb{P}(T < c_{W}(\alpha) ) - \alpha|\leq C n^{-c},$$
where $c_{W}(\alpha)$ is the $\alpha$-quantile of $W$ conditional on $(x_i)_{i=1}^n$ given by
$$c_{W}(\alpha) := \inf\{t\in \mathbb R: \mathbb{P}_e(W \leq t) \geq \alpha\}.$$

\end{thm}

\subsection{Oracle inequalities for logistic regression under misspecification}
\label{subsec:orac}

We require a condition on the design matrix known as the compatibility condition \citep{hds}.

\begin{defn}[Compatibility constant]
\label{cc}
We say that the compatibility condition is met with constant $\phi > 0$ if for all $\beta\in\mathbb R^p$ that satisfy
$\|\beta_{S^c}\|_1 \leq 3\|\beta_S\|_1,$
it holds that
$$\|\beta_S\|_1^2 \leq \frac{s\|X\beta\|_2^2}{\phi^2}.$$ 
\end{defn}
%


For our final lemma, we use the notation of Sections~\ref{sec:cons.logistic} and~\ref{Sec:LogisticProofs}.

\begin{lem}
\label{logit.rates}
Suppose that there exists a constant $K>0,$ such that  
$$\max_{1\leq j\leq p}\|X_j\|_\infty \leq K, \;\;\max_{1\leq j\leq p}\|X_j\|_2\leq 1,$$
and 
that $X$ satisfies the compatibility condition with constant $\phi>0.$
Take $t>0$ and let 
$$\bar\lambda:= \sqrt{2\log(2p)/n} + K\log (2p)/(3n),$$
$$\lambda_0:=4\bar\lambda+tK/(3n) + \sqrt{2t(1+8\bar\lambda)/n}.$$
Assume that there exist constants $\epsilon_0,\eta$ such that 
$$0< \eta \leq \epsilon_0 < \pi_0(x) < 1-\epsilon_0,\;\;\text{ for all } \|x\|_\infty \leq K,$$ 
$$\sup_{x:\|x\|_\infty \leq K} |f_{\beta_0}(x) - f_0(x)|\leq \eta/2.$$
For some constant $M \geq 8$ take $\lambda$ satisfying  $8\lambda_0\leq \lambda \leq M\lambda_0$ and $17\lambda s(e^\eta / \epsilon_0 + 1)^2 / \phi^2 \leq \eta/(2K)$, and further assume that 
\begin{align*}
R( f_{\beta_0}|X) - R(f_0|X)   &\leq \min\left\{\eta\lambda_0/4, \frac{\lambda^2{s}(e^{\eta}/\epsilon_0+1)^2}{ 6\phi^2}\right\}, \\
  \frac{8KM^2 (e^{\eta}/\epsilon_0+1)^2}{\eta}\frac{\lambda_0 s}{\phi^2} &\leq 1.
\end{align*}                                                                   
Then with probability at least $1-e^{-t},$
it holds that
$$R(f_{\hat{\beta}}|X) - R(f_0|X) + \lambda \|\hat\beta-\beta_{0}\|_1 \leq \frac{17\lambda^2{s}(e^{\eta}/\epsilon_0+1)^2}{ \phi^2}.$$

\end{lem}

\begin{proof}[Proof of Lemma \ref{logit.rates}]
The proof follows from Lemma 6.8 and Section 6.7 in \cite{hds}.
\end{proof}

%
%
%
%
%
%